\documentclass[11pt]{article} 

\usepackage[utf8]{inputenc} 

\usepackage{geometry} 
\usepackage{graphicx} 
\usepackage{color}

\definecolor{Red}{rgb}{1,0,0}
\definecolor{Blue}{rgb}{0,0,1}
\definecolor{Olive}{rgb}{0.41,0.55,0.13}
\definecolor{Green}{rgb}{0,1,0}
\definecolor{MGreen}{rgb}{0,0.8,0}
\definecolor{DGreen}{rgb}{0,0.55,0}
\definecolor{Yellow}{rgb}{1,1,0}
\definecolor{Cyan}{rgb}{0,1,1}
\definecolor{Magenta}{rgb}{1,0,1}
\definecolor{Orange}{rgb}{1,.5,0}
\definecolor{Violet}{rgb}{.5,0,.5}
\definecolor{Purple}{rgb}{.75,0,.25}
\definecolor{Brown}{rgb}{.75,.5,.25}
\definecolor{Grey}{rgb}{.5,.5,.5}

\def\red{\color{Red}}

\usepackage{booktabs} 
\usepackage{array} 
\usepackage{paralist} 
\usepackage{verbatim} 
\usepackage{subfigure} 
\usepackage{amsmath,amssymb,amsthm,mathrsfs}
\usepackage[colorlinks]{hyperref}
\usepackage{blkarray}

\usepackage{fancyhdr} 
\makeatletter
\newtheorem*{rep@theorem}{\rep@title}
\newcommand{\newreptheorem}[2]{%
\newenvironment{rep#1}[1]{%
 \def\rep@title{#2 \ref{##1}}%
 \begin{rep@theorem}}%
 {\end{rep@theorem}}}
\makeatother

\theoremstyle{plain}

\newtheorem{theorem}{Theorem}[section] 
 
\newtheorem{corollary}{Corollary}[section]

\newtheorem{lemma}{Lemma}[section]

\newtheorem{theorem*}{Theorem}   
\newtheorem{lemma*}{Lemma} 
\newtheorem{corollary*}{Corollary} 
\newtheorem*{remark*}{Remark}
\newtheorem{example}{Example}

\newtheorem{remark}{Remark}[section]

\newlength{\widebarargwidth}
\newlength{\widebarargheight}
\newlength{\widebarargdepth}

\theoremstyle{definition}
\newtheorem{definition}{Definition}[section]

\def\cC{{\cal C}}
\def\cD{{\cal D}}

\def\cF{{\cal F}}
\def\cG{{\cal G}}

\def\cI{{\cal I}}

\def\cN{{\cal N}}

\def\cT{{\cal T}}

\allowdisplaybreaks 

\begin{document}

\begin{center}

{\bf{\Large{Intrinsic entropies of log-concave distributions}}}

\vspace*{.25in}

\begin{tabular}{ccc}
{\large{Varun Jog}} & \hspace*{.75in} & {\large{Venkat Anantharam}} \\ 
{\large{\texttt{vjog@wisc.edu}}} & & {\large{\texttt{ananth@eecs.berkeley.edu}}} \vspace{.2in}
 \\
ECE Department && EECS Department \\
University of Wisconsin - Madison && University of California, Berkeley \\ WI 53706 & & CA 94720
\end{tabular}

\vspace*{.2in}

\date{\today}

\vspace*{.2in}

\end{center}


\begin{abstract}
The entropy of a random variable is well-known to equal the exponential growth rate of the volumes of its typical sets. In this paper, we show that for any log-concave random variable $X$, the sequence of the $\lfloor n\theta \rfloor^{\text{th}}$ intrinsic volumes of the typical sets of $X$ in dimensions $n \geq 1$ grows exponentially with a well-defined rate. We denote this rate by $h_X(\theta)$, and call it the $\theta^{\text{th}}$ intrinsic entropy of $X$. We show that $h_X(\theta)$ is a continuous function of $\theta$ over the range $[0,1]$, thereby providing a smooth interpolation between the values  0 and $h(X)$ at the endpoints 0 and 1, respectively.\\

\end{abstract}
\section{Introduction}


Convex geometry and log-concave functional analysis have been studied in tandem for many decades. Log-concave functions (or measures) form a convenient stepping stone between geometry and analysis --- the indicator function of a convex set is log-concave, so one may study convex sets by studying the larger class of log-concave measures. Furthermore, log-concave measures enjoy many advantages over convex sets, including closure under convolution, making them amenable to analysis. Ball \cite{Bal86} carried out the first geometric study of log-concave functions, connecting isotropic measures to isotropic convex sets and generalizing several convex geometric inequalities to log-concave measures. The idea of extending and reinterpreting convex geometric results in terms of log-concave measures has since then taken root, and the term ``geometrization of probability (or analysis)"  coined by Milman \cite{Mil07} has been used to describe this approach. Milman \cite{Mil07} and Klartag \& Milman \cite{Kla05} noted that developing a structured theory of the geometry of log-concave measures could have important consequences for problems that are purely geometric in nature \cite{Kla05, BobMad10}. 

In a recent survey, Madiman et al.\ \cite{MadEtAl16} identified two approaches that are employed to geometrize analysis: functional (or integral) lifting and probabilistic (or entropic) lifting. To illustrate these liftings, we first state the famous Brunn-Minkowski inequality in convex geometry \cite{gardner2002brunn}:
\begin{theorem}[Brunn-Minkowski inequality for compact convex sets]
Let $0 < \lambda < 1$, and let $X$ and $Y$ be compact convex sets. Let $|X|$ and $|Y|$ denote the $n$-dimensional volumes of $X$ and $Y$. Define the Minkowski sum of $X$ and $Y$ as 
$$X \oplus Y := \{x + y | x \in X, y \in Y\}.$$
Then the following equivalent statements hold:
\begin{enumerate}
\item[(i)] $|(1-\lambda)X \oplus \lambda Y| \geq |X|^{1-\lambda} |Y|^\lambda$, \text{ and} 
\item[(ii)] $|(1-\lambda)X  \oplus \lambda Y|^{1/n} \geq (1-\lambda)|X|^{1/n} + \lambda |Y|^{1/n}.$
\end{enumerate}
\end{theorem}
Here and in the rest of the paper $:=$ denotes equality by definition. Integral and entropic lifting can be understood with regards to the Brunn-Minkowski inequality as follows:

\begin{enumerate}
\item
\textbf{Integral lifting:} Madiman et al.\ \cite{MadEtAl16} classify this approach as one that replaces convex sets by log-concave functions and replaces the volume functional by the integral. The Pr\'{e}kopa-Leindler inequality \cite{Pre71, Lei72, Pre73, gardner2002brunn}, which serves as an integral lifting of the Brunn-Minkowski inequality (BMI), as stated in (i):
\begin{theorem}[Pr\'{e}kopa-Leindler inequality]
Let $0 \leq \lambda < 1$, and let $f, g,$ and $h$ be non-negative integrable functions on $\mathbb R^n$ satisfying
$$h((1-\lambda) x + \lambda y) \geq f(x)^{1- \lambda}g(y)^{\lambda}.$$
Then
$$\int_{\mathbb R^n} h(x) dx \geq \left(\int_{\mathbb R^n} f(x) dx \right)^{1-\lambda} \left( \int_{\mathbb R^n} g(x) dx\right)^{1-\lambda}.$$
\end{theorem}
Apart from visual similarity, this integral lifting indeed provides a generlization; i.e., one can prove the Brunn-Minknowski inequality using the Pr\'{e}kopa-Leindler inequality by replacing the functions $f$, $g$, and $h$ by $ 1_{X}$, $1_{Y}$, and $1_{(1-\lambda) X + \lambda Y}$, respectively.
\item
\textbf{Entropic lifting:} As per Madiman et al.\ \cite{MadEtAl16}, one replaces convex sets by random variables and the volume functional by entropy. The entropy power inequality (EPI) \cite{shannon, Sta59, Bla65} serves as an entropic lifting of the Brunn-Minkowski inequality:
\begin{theorem}[Entropy power inequality]
Let $X$ and $Y$ be random variables on $\mathbb R^n$ with well-defined entropies $h(X)$ and $h(Y)$. Then the following inequality holds:
$$ e^{2h(X+Y)/n} \geq e^{2h(X)/n} + e^{2h(Y)/n}.$$
\end{theorem}
Here, we may observe a striking similarity to the Brunn-Minkowski inequality (ii). Unlike the integral generalization, however, the EPI is not readily proven using the BMI, or vice versa, and neither can be said to be a generalization of the other.
\end{enumerate}

Our work in this paper decidedly has a ``geometrization or probability" flavor, and although it does not fall squarely in either of the above classifications, it is closer in spirit to the approach of entropic lifting of convex geometry.  Several instances of entropic liftings have been studied in the literature, and we provide a few examples to highlight the scope of such liftings:
\begin{enumerate}
\item
\textbf{Surface area and Fisher information:} The surface area of a compact, convex set $X \subseteq \mathbb R^n$ is defined as the limit
$$\partial(X) := \lim_{\epsilon \to 0_+} \frac{|X \oplus \epsilon B| - |X|}{\epsilon} = \frac{d}{d\epsilon} |X \oplus \epsilon B| \Big |_0,$$
where $B$ is the Euclidean ball of unit radius. To obtain an entropic lifting, one may replace the convex set by a random variable $X$ and replace volume by entropy. Instead of a Minkowski sum with an $\epsilon$-ball, we use a Gaussian random variable with variance $\epsilon$:
$$\lim_{\epsilon \to 0_+} \frac{h(X+\sqrt \epsilon Z) - h(X)}{\epsilon} = \frac{d}{d\epsilon} h(X+\sqrt \epsilon Z)\Big|_0 \stackrel{(a)}= \frac{1}{2}J(X),$$
where $Z \sim \cN(0,1)$ and $J(X)$ is the Fisher information of $X$. The equality in step $(a)$ is known in the literature as De Bruijn's identity \cite{Sta59, cover1991}. This indicates that the entropic lifting of surface area is Fisher information.
\item
\textbf{Isoperimetric inequality:} Among all compact convex shapes with a fixed surface area, the Euclidean ball has the largest volume \cite{schneider2013convex}. The entropic lifting of this result---that among all distributions with a fixed Fisher information, the Gaussian distribution has the maximum entropy---holds \cite{cover1991}.

\item
\textbf{Concavity of entropy power:} If $X \subseteq \mathbb R^n$ is a compact, convex set, the function $f(t) := |X+tB|^{1/n}$ is concave on $t \geq 0$. This follows from a straightforward application of the Brunn-Minkowski inequality. The entropic lifting is known as Costa's EPI \cite{Cos85, Dem89, Vil00}, or concavity of entropy power, and states the following: If $X$ is a random variable on $\mathbb R^n$ and $Z \sim \cN(0, I)$, then the function $e^{2h(X+\sqrt t Z)/n}$ is concave on $t \geq 0$.

\item
\textbf{Reverse BMI and EPI:} The reverse BMI discovered by Milman \cite{Mil86} states that given two convex sets $X$ and $Y$, volume-preserving linear transformations exist mapping $X$ to $\tilde X$ and $Y$ to $\tilde Y$, such that the following inequality holds:
$$|\tilde X \oplus \tilde Y| \leq C(|\tilde X|^{1/n} + |\tilde Y|^{1/n}),$$
where $C$ is an absolute constant. Bobkov \& Madiman \cite{BobMad12} showed that for log-concave (or more generally $\kappa$-concave, see \cite{BobMad12} for a definition) random variables $X$ and $Y$, entropy-preserving linear transformations exist mapping random variables $X$ and $Y$ to $\tilde X$ and $\tilde Y$, such that
$$e^{2h(\tilde X+ \tilde Y)/n} \leq C\left( e^{2h(\tilde X)/n} + e^{2h(\tilde Y)/n}\right),$$
for an absolute constant $C$.
\end{enumerate}

Our work in this paper concerns the entropic liftings of what are called \emph{intrinsic volumes} in convex geometry. This is closely related to the liftings 1 and 2 mentioned above. In 1, it is natural to examine higher-order derivates with respect to $\epsilon$ for further parallels. A remarkable result in convex geometry, known as Steiner's formula \cite{schneider2013convex, klainrota}, shows that all sufficiently high-order derivatives of the volume functional are equal to 0. Since no such property holds for the derivatives of entropy, it places a limition on how far we may extend the analogy. Steiner's formula is stated below:
\begin{theorem}[Steiner's formula  \cite{schneider2013convex, klainrota}]
Let $X \subseteq \mathbb R^n$ be a compact, convex set. Let $B_j$ denote the $j$-dimensional Euclidean ball in $\mathbb R^n$, and let $\omega_j = |B_j|$. Then we have the following equality:
$$|X \oplus tB_n| = \sum_{j=0}^n V_{n-j}(X) \omega_j t^j,$$
where $\{V_0(X), V_1(X), \dots, V_n(X)\}$ are the $n+1$ intrinsic volumes of $X$.
\end{theorem}
 Intrinsic volumes are functions defined on the class of compact, convex sets, and can be uniquely extended to \emph{polyconvex sets}; i.e., sets that are finite unions of compact, convex sets. Some of these intrinsic volumes are known in the literature under alternate names (e.g., $V_0(X)$ is the Euler characteristic, $V_1(X)$ is the mean width, $2V_{n-1}(X)$ is the surface area, and $V_n(X)$ is the volume). Intrinsic volumes have a number of interpretations in geometry. We state some of these interpretations, as found in Klain \& Rota \cite{klainrota}. Intrinsic volumes are \emph{valuations} on polyconvex sets; i.e., for all polyconvex sets $X$ and $Y$ and for all $0 \leq i \leq n$,
\begin{equation*}
V_i(X \cup Y) = V_i(X) + V_i(Y) - V_i(X\cap Y).
\end{equation*}
Furthermore, these valuations are convex-continuous and  invariant under rigid motions \cite{schneider2013convex}. In fact, Hadwiger's theorem \cite{Had13, Kla95, klainrota} states that any convex-continuous, rigid-motion invariant valuation on the set of polyconvex sets is a linear combination of the intrinsic volume valuations. Here, convex-continuity is defined with respect to the topology on compact convex sets induced by the Hausdorff metric $\delta$, which measures the distance between $X$ and $Y$ according to the relation
$$\delta(X,Y) := \max \left\{\sup_{x \in X} \inf_{y\in Y} ||x-y||, \; \sup_{y \in Y} \inf_{x\in X} ||x-y||\right\}.$$
Kubota's theorem or Crofton's formula \cite{klainrota, schneider2013convex} implies that the $i^\text{th}$ intrinsic volume $V_i(K)$ is proportional to the volume of a random $i$-dimensional projection or slice of $K$. Intrinsic volumes are thus defined by the geometric structure of a set and describe its global characteristics. 

A number of inequalities concerning volumes carry over for intrinsic volumes. Schneider \cite{schneider2013convex} notes that a \emph{complete}-Brunn-Minkowski inequality exists for intrinsic volumes of compact convex sets:\begin{equation*}
V_i(X\oplus Y)^{1/i} \geq V_i(X)^{1/i} + V_i(Y)^{1/i}, \text{ for all } i \geq 1.
\end{equation*}
Furthermore, the isoperimetric inequality from point 2 holds much more generally \cite{schneider2013convex}: If $B$ is the unit Euclidean ball in $\mathbb R^n$, 
\begin{equation*}
\left( \frac{V_n(X)}{V_n(B)}\right)^{1/n}  \leq \left( \frac{V_j(X)}{V_j(B)}\right)^{1/j}. 
\end{equation*}
These results pose a natural question, which is also central to our work: What is a natural analog of intrinsic volumes for log-concave measures? In other words, what is an appropriate functional or probabilistic lifting of intrinsic volumes?

Several recent papers have proposed functional analogs of intrinsic volumes. In Milman \& Rotem \cite{MilRot13b, MilRot13a} and Bobkov et al.\ \cite{BobEtAl14} the authors introduce a new notion of addition and interpret intrinsic volumes (or more generally, mixed volumes) by considering the coefficients of the version of Steiner's formula corresponding to this notion of addition. For log-concave functions $f$ and $g$ on $\mathbb R^n$, the addition operation $f \stackrel{\sim}+ g$ is given as
\begin{equation}
(f \stackrel{\sim}+ g) (z) := \sup \left\{ \min_{(x,y) \text{ s.t. } x+y = z} \{f(x), \; g(y)\} \right\}.
\end{equation}
The authors obtain analogs of several inequalities, including the Brunn-Minkowski inequality, for intrinsic volumes defined with respect to this operation. We refer the reader to the recent survey by Colesanti \cite{Col16} for further references regarding such functional liftings, as well as an alternate approach to generalizing intrinsic volumes by searching for appropriate notions of valuations on the space of log-concave functions.

In this paper, we attempt to discover an entropic lifting of intrinsic volumes for log-concave distributions. To connect a probability measure to its entropy, we use the well-known result from information theory that connects the volume of a typical set of a random variable to its entropy \cite{cover1991}. To be precise, given a real-valued random variable $X$ with density $p_X$ and differential entropy $h(X)$, one way to define its $\epsilon$-typical set, $\hat \cT^\epsilon_n$ in dimension $n$ is 
\begin{equation}\label{eq: def Tn 0}
\hat \cT^\epsilon_n := \left\{ x^n \in \mathbb R^n \! \mid \! e^{-n(h(X)+ \epsilon)} \leq p_{X^n}(x^n)  \leq e^{-n(h(X)-\epsilon)}\right\},
\end{equation}
where $p_{X^n}(x^n) = \prod_{i=1}^n p_{X}(x_i)$. For all sufficiently large $n$, the volume $|\hat \cT^\epsilon_n|$, satisfies \cite{cover1991}
\begin{equation}\label{eq: Tn hx}
(1-\epsilon)e^{n(h(X)-\epsilon)} \leq |\hat \cT^\epsilon_n| \leq e^{n(h(X)+\epsilon)}.
\end{equation}
Thus, the exponential growth rate of  the volume $|\hat \cT^\epsilon_n|$ is determined by the differential entropy $h(X)$. Extending this analogy, we postulate that the intrinsic volumes of $\{\hat \cT^\epsilon_n \}_{n \geq 1}$, or the exponential growth rate of the same sequence, yields analogs of intrinsic volumes for distributions. To ensure that typical sets have well-defined intrinsic volumes, we focus our attention on the one-sided typical sets of log-concave random variables. We formally define log-concave distributions and their typical sets, as follows:

\begin{definition}[Log-concave distribution]
The distribution $p_X$ of a real-valued random variable $X$ is said be log-concave if there exists a lower-semicontinuous convex function $\Phi: \mathbb R \to \mathbb R \cup \{ + \infty \}$ such that $p_X(x) = e^{-\Phi(x)}$ for all $x \in \mathbb R$.
\end{definition}

\begin{definition}\label{def: typical}
Let $X$ be a real-valued random variable with a log-concave density $p_X(X) := e^{-\Phi(x)}$. For each $n \geq 1$ and $\epsilon > 0$, the one-sided $\epsilon$-typical set is given by
\begin{align}
\cT^\epsilon_n &:= \left\{x^n \in \mathbb R^n ~|~ p_{X^n}(x^n) \geq e^{-n(h(X)+\epsilon)}\right\} \label{eq: def Tn 1}\\
&= \left\{x^n \in \mathbb R^n ~|~ \sum_{i=1}^n \Phi (x_i) \leq {n(h(X)+\epsilon)}\right\}.\label{eq: def Tn 4}
\end{align}
\end{definition}

\begin{remark}
For a lower semi-continuous function $\Phi$, it is possible to define a closed set as the support of $p_X$, such that $p_X$ is continuous on this support. The restriction to lower-semicontinuous functions is for technical reasons and is not crucial to our analysis. For example, the function defined by
\begin{align*}
\Phi(x) = 
\begin{cases}
0 &\text{ if } x \in (0,1),\\
+\infty &\text{ otherwise}, 
\end{cases}
\end{align*}
corresponding to the uniform distribution on $(0,1)$, is not a lower semi-continuous convex function. However, the typical sets of a uniform distribution on $(0,1)$ are essentially the same as those of a uniform distribution on $[0,1]$ Thus, we may consider the lower-semicontinuous version of the function
\begin{align*}
\Phi(x) = 
\begin{cases}
0 &\text{ if } x \in [0,1],\\
+\infty &\text{ otherwise} 
\end{cases}
\end{align*}
in our analysis of typical sets.
\end{remark}

Since $\cT^\epsilon_n$ is a level set of the lower semi-continuous function $\sum_{i=1}^n \Phi(x_i)$, it is immediate that typical sets are compact and convex. Importantly, this implies they have well-defined intrinsic volumes. Denote the intrinsic volumes of $\cT^\epsilon_n$ by $\{\mu^\epsilon_n(0), \dots, \mu^\epsilon_n(n)\}$. As noted earlier, the $n^\text{th}$ intrinsic volume is simply the volume, and its exponential growth rate is determined by the differential entropy, since inequality \eqref{eq: Tn hx} continues to hold for $\cT^\epsilon_n$. In other words,
\begin{equation}\label{eq: nose}
h(X) = \lim_{\epsilon \to 0_+} \lim_{n \to \infty} \frac{1}{n} \log \mu^\epsilon_n(n).
\end{equation}
Our analog of intrinsic volumes is obtained by taking the limit
\begin{equation}\label{eq: mu theta}
h_\theta(X) := \lim_{\epsilon \to 0_+} \lim_{n \to \infty} \frac{1}{n} \log \mu^\epsilon_n({\lfloor n\theta \rfloor}),
\end{equation}
for $\theta \in [0,1]$. We refer to this function as $h_\theta(X)$ or $h_X(\theta)$, based on which parameter is considered to be fixed in the specific context. Note that unlike earlier works, this yields an entire continuum of ``intrinsic volumes," which we also refer to as \emph{intrinsic entropies}. For $\theta = 0$ and $1$, we observe that $h_0(X) = 0$ and $h_1(X) = h(X)$. For values of $\theta \in (0,1)$, the existence of the limit as defined in equation \eqref{eq: nose} is far from obvious. For some special cases, however, we can evaluate intrinsic entropies for all values of $\theta$:
\begin{example}
Let $X \sim \cN(0, \nu)$. The one-sided $\epsilon$-typical set in this case is simply the $n$-dimensional ball of radius $n\nu(1+2\epsilon)$, denoted by $B_n\left(\sqrt{n\nu(1+2\epsilon)}\right)$. The intrinsic volumes of such a ball admit a closed-form expression \cite{klainrota}, and the $j^\text{th}$ intrinsic volume is given by 
\begin{align}
V_j\left(B_n(\sqrt{n\nu(1+2\epsilon)})\right) = {n \choose j} \frac{\omega_n}{\omega_{n-j}}\left(n\nu(1+2\epsilon)\right)^{j/2},
\end{align}
where $\omega_i$ is the volume of the $i$-dimensional unit ball. Substituting $j = \lfloor n\theta \rfloor$ and taking the desired limits yields
\begin{align}
h_X(\theta) =  H(\theta) + \frac{\theta}{2}\log 2\pi e\nu + \frac{1-\theta}{2}\log(1-\theta),
\end{align}
where $H(\theta) = -\theta\log\theta - (1-\theta)\log(1-\theta)$ is the binary entropy function.
\end{example}
\begin{example}
Let $X$ be a random variable distributed uniformly in the interval $[0, A]$. For all $\epsilon > 0$, the one-sided $\epsilon$-typical set for $X$ is the $n$-dimensional cube $[0, A]^n$. The $j^\text{th}$ intrinsic volume of this cube \cite{klainrota} is given by
\begin{align}
V_j\left([0,A]^n\right) = {n \choose j}A^j.
\end{align}
Substituting $j = \lfloor n\theta \rfloor$ and taking the desired limits gives
\begin{align}
h_X(\theta) = H(\theta) + \theta \log A.
\end{align}
\end{example}
Note that in the above examples, not only does the limit exist for all values of $\theta \in [0,1]$, but it is also continuous as a function of $\theta$. For an arbitrary log-concave distribution, such an explicit calculation is not possible, as the intrinsic volumes of its typical sets are not available in closed form. Even so, the main result of this paper establishes that similar properties hold for all log-concave distributions:
\begin{theorem}\label{thm: main}
Let $X \sim p_X$ be a log-concave random variable. Then there exists a continuous function $h_X : [0,1] \to \mathbb R$ such that for all $\theta \in [0, 1]$, 
\begin{equation}
h_X(\theta) := \lim_{\epsilon \to 0_+} \lim_{n \to \infty} \frac{1}{n} \log \mu^\epsilon_n({\lfloor n\theta \rfloor}).
\end{equation}
\end{theorem}
The remainder of this paper is structured as follows: In Section \ref{section: candidate}, we produce a candidate function $-\Lambda^*$, which is a continuous function on $[0,1]$, and show in the subsequent sections that it equals $h_X$. We break up the proof of Theorem \ref{thm: main} into three parts, A, B, and C as follows: In Section \ref{section: part a}, we show part A, that $h_X(\theta) = -\Lambda^*(\theta)$ for $\theta \in (0,1)$. In Section \ref{section: part b}, we show part B, that $h_X(0) = -\Lambda^*(0)$. In Section \ref{section: part c}, we show part C, that $h_X(1) = -\Lambda^*(1)$. We conclude the paper in Section \ref{section: discussion} by discussing future work and open problems.


\section{A candidate function}\label{section: candidate}

Let $X$ be a real-valued random variable with a log-concave density $p_X(X)$, given by  $p_X(x) = e^{-\Phi(x)}$, for a convex function $\Phi: \mathbb R \to \mathbb R \cup {+\infty}$. Our first lemma establishes a ``super-multiplicative" property of the sequence of typical sets as defined in equation \eqref{eq: def Tn 1}:
\begin{lemma}\label{lemma: superconv Tn}
The sequence of sets $\{\cT^\epsilon_n\}_{n\geq 1}$  satisfies
\begin{equation}\label{eq: superconv Tn}
\cT^\epsilon_m \times \cT^\epsilon_n \subseteq \cT^\epsilon_{m+n}, \text{~~for all~~} m,n \geq 1.
\end{equation}
\end{lemma}
\begin{proof}
Let $x^m \in \cT^\epsilon_m$ and $y^n \in \cT^\epsilon_n$. We have
\begin{align*}
\sum_{i=1}^m \Phi (x_i) \leq m(h(X)+\epsilon) \text{~~ and~~} \sum_{i=1}^n \Phi (y_i) \leq n(h(X)+\epsilon).
\end{align*}
Adding the above inequalities, $z^{m+n} = (x^m, y^n)$ satisfies
\begin{align*}
\sum_{i=1}^{m+n} \Phi (z_i) \leq (m+n)(h(X)+\epsilon),
\end{align*}
which implies that $z^{m+n} \in \cT^\epsilon_{m+n}$.
\end{proof}

This super-multiplicative growth of typical sets is a \emph{geometric} property of typical sets. Therefore, it is unsurprising that one may glean some information about the intrinsic volumes of typical sets from such a geometric property. We state the key results that we need, the full details of which may be found in Klain \& Rota \cite{klainrota}:
\begin{theorem}[Properties of intrinsic volumes]\label{thm: iv props}
Let $X$ and $Y$ be any compact, convex sets. Then the following statements hold:
\begin{enumerate}
\item[(a)]
If $X \subseteq Y$, then $V_j(X) \leq V_j(Y)$ for all $j \geq 0$; i.e.,  intrinsic volumes are monotonic with respect to inclusion.
\item[(b)]
The intrinsic volumes of the Cartesian product $X \times Y$ are obtained by convolving the intrinsic volumes of $X$ and $Y$; i.e., for all $j \geq 0$,
\begin{align}
V_j(X \times Y) = \sum_{i=0}^j V_i(X)V_{j-i}(Y).
\end{align}
\end{enumerate}
\end{theorem}

Lemma \ref{lemma: superconv Tn} and Theorem \ref{thm: iv props} immediately imply the following result:

\begin{lemma}\label{lemma: superconv mu}
For $\epsilon > 0$, let the one-sided $\epsilon$-typical sets of a log-concave random variable be $\{\cT^\epsilon_n\}_{n \geq 1}$, and let the sequence of intrinsic volumes of these typical sets be $\{\mu^\epsilon_n(\cdot)\}_{n \geq 1}$. Then the sequence of intrinsic volumes  satisfies the following:
\begin{equation}\label{eq: superconv mu}
(\mu^\epsilon_m \star \mu^\epsilon_n)(i) \leq \mu^\epsilon_{m+n}(i), ~~\forall m, n\geq 1, \text{ and } \forall i \geq 0,
\end{equation}
where $``\star"$ stands for function convolution.
\end{lemma}
\begin{remark}
Note that in the above statement, we define $\mu^\epsilon_n(j) = 0$ for $j \geq n+1$ to obtain the function $\mu^\epsilon_n: \mathbb Z_+ \to \mathbb R$ starting from the finite sequence $\{\mu^\epsilon_n(0), \dots, \mu^\epsilon_n(n)\}$. This is also geometrically accurate, since higher-dimensional intrinsic volumes of a lower-dimensional set are equal to 0.
\end{remark}

The analysis of sequences that satisfy property \eqref{eq: superconv mu} is central to our work, and we devote the next subsection to this topic.

\subsection{Super-convolutive sequences}\label{subsection: superconv}
We define sequences that satisfy a property such as inequality \eqref{eq: superconv mu} as \emph{super-convolutive sequences}:

\begin{definition} \label{def: supconv}
For each $n \geq 1$, let $\mu_n: \mathbb Z_+ \to \mathbb R_+$ be such that $\mu_n(i) = 0$, for all $i \geq n + 1$. The sequence $\{\mu_n(\cdot)\}_{n\geq 1}$ is super-convolutive if
\begin{equation}\label{eq: def superconv mu}
(\mu_m \star \mu_n) (i)\leq \mu_{m+n}(i), ~~\forall m, n \geq 1, \quad \text{and}~~ \forall i \geq 0.
\end{equation}
\end{definition}

We now turn to the analysis of super-convolutive sequences, with an emphasis on their convergence properties. First, we develop some notation and make some small observations:

\begin{lemma}\label{lemma: lambda Tn}
Let $\{\mu_n(\cdot)\}$ be as in Definition \ref{def: supconv}. Define the sequence of generating functions $\{G_n\}_{n\geq 1}$ and $\{g_n\}_{n \geq 1}$ as follows:
\begin{equation}
G_n(t) := \log \sum_{j=0}^n \mu_n(j)e^{jt}, \text{~~and~~} g_n(t) := \frac{G_n(t)}{n}.
\end{equation}
The following results hold:
\begin{enumerate}
\item[(i)] 
For all $m, n \geq 1$ and for all $t$, 
\begin{equation}\label{eq: supc condition on G} 
G_m(t) + G_n(t) \leq G_{m+n}(t).
\end{equation}

\item[(ii)]
For all $t$, the following limit, denoted by $\Lambda(t)$, exists (although it may be $+\infty$):
\begin{equation}\label{eq: def: lambda}
\Lambda(t) := \lim_{n \to \infty} g_n(t).
\end{equation}
\end{enumerate}
\end{lemma}
\begin{proof}[Proof of Lemma \ref{lemma: lambda Tn}]
Statement (i) is immediate from inequality \eqref{eq: def superconv mu}, whereas (ii) follows from a direct application of the super-additive convergence theorem, also known as Fekete's Lemma \cite{steele1997probability}.
\end{proof}

A super-convolutive sequence $\{\mu_n\}$, without any other conditions imposed on it apart from condition \eqref{eq: def superconv mu}, can blow up to infinity arbitrarily fast. To prevent such scenarios, we define a \emph{proper super-convolutive sequence} as follows:

\begin{definition}\label{def: prop supconv}
A super-convolutive sequence as in Definition \ref{def: supconv} is said to be \emph{proper} if it satisfies the following conditions:
\begin{enumerate}
\item[(i)] For all $n$, we have $\mu_n(n) >0$ and $\mu_n(0)>0$. 

\item[(ii)] The limit $\beta := \lim_{n \to \infty} \frac{1}{n} \log \mu_n(0)$ is finite.

\item[(iii)] The value $\gamma := \Lambda(0)$ is finite.

\end{enumerate}
\end{definition}
\begin{remark}\label{remark: prop supconv}
Note that $(\mu_m \star \mu_n)(0) = \mu_n(0)\mu_m(0)$. Thus, the limit in condition (ii) is assured to exist by super-additivity. Note also that $\Lambda(0)$ has a simple expression:
$$\Lambda(0) = \lim_{n \to \infty} \frac{1}{n} \log \left( \sum_{i=0}^n \mu_n(i)\right).$$ The limit
$$\alpha := \lim_{n \to \infty} \frac{1}{n} \log \mu_n(n)$$
is also assured to exist, since $\mu_{m+n}(m+n) = \mu_m(m)\mu_n(n)$, and condition (iii) ensures that it is finite.
\end{remark}

The limit function $\Lambda$ of a proper super-convolutive sequence satisfies a number of desirable properties, including $\Lambda(t) < \infty$, for all $t$. We gather some of these properties in the following lemma:

\begin{lemma}[Proof in Appendix \ref{proof: lemma: lambda props}]\label{lemma: lambda props}
For a proper super-convolutive sequence, the limit function $\Lambda$, as defined in equation \eqref{eq: def: lambda}, satisfies the following properties:
\begin{enumerate}
\item[(i)] $g_1(t) \leq \Lambda(t) \leq \max(\gamma, t +\gamma)$, for all $t$.

\item[(ii)] $\Lambda$ is convex and monotonically increasing.
\end{enumerate}
\end{lemma}

Our next lemma details some important properties of $\Lambda^*$, the convex conjugate \cite{boyd2009convex} of $\Lambda$:
\begin{lemma}[Proof in Appendix \ref{proof: lemma: lambda star props}] \label{lemma: lambda star props}
Let $\Lambda$ be as in Lemma \ref{lemma: lambda props}, and let $\Lambda^*$ be its convex conjugate. Then $\Lambda^*$ satisfies the following properties:
\begin{enumerate}
\item[(i)]
The domain of $\Lambda^*$ is $[0,1]$, and $\Lambda^*$ is continuous and convex on its domain.
\item[(ii)]
For $t \not\in \{0,1\}$, the value of $\Lambda^*(t)$ is given by the limit
\begin{equation}
\Lambda^*(t) = \lim_{n \to \infty} g_n^*(t).
\end{equation}
\item[(iii)] For $t \in \{0, 1\}$, we have the inequality
\begin{equation}
\Lambda^*(t) \leq  \lim_{n \to \infty} g_n^*(t).
\end{equation}
\end{enumerate}
\end{lemma}
\begin{remark}\label{remark: lambda star props}
It is easy to check that the domain of $g_n^*$ is $[0,1]$ and that
$$\lim_{n \to \infty} g_n^*(0) = -\lim_{n\to \infty} \frac{1}{n} \log \mu_n(0),$$
and 
$$\lim_{n \to \infty} g_n^*(1) = -\lim_{n\to \infty} \frac{1}{n} \log \mu_n(n),$$
and that $g_n^*(x) \ge -\gamma$ for all $x \in [0,1]$. Using the notation from Definition \ref{def: prop supconv} and Remark \ref{remark: prop supconv}, statement (iii) is equivalent to saying that $\Lambda^*(0) \leq -\beta$ and $\Lambda^*(1) \leq -\alpha.$ We also note that it is possible to obtain a strict inequality statement (ii), and an example is provided in Appendix \ref{example: strict inequality}.
\end{remark}

We are now in a position to state our main result concerning convergence of proper super-convolutive sequences:

\begin{theorem}\label{thm: lambda Tn}	
Let $\{\mu_n\}$ be as in Definition \ref{def: prop supconv}. Define a sequence of measures $\{\mu_{n/n}\}_{n\geq 1}$ supported on $[0,1]$ by assigning point masses as follows:
$$\mu_{n/n}\left(\left\{\frac{j}{n}\right\}\right) := \mu_n(j),  \text{~~for~~}  0 \leq j \leq n.$$
 Let $I \subseteq \mathbb{R}$ be a closed set and $F \subseteq \mathbb{R}$ be an open set. Then 
\begin{align}
\limsup_{n \to \infty} \frac{1}{n} \log \mu_{n/n}(I) &\leq - \inf_{x \in I} \Lambda^*(x), \qquad \text{and}\label{eq: UBB}\\
\liminf_{n \to \infty} \frac{1}{n} \log \mu_{n/n}(F) &\geq - \inf_{x \in F} \Lambda^*(x)\label{eq: LBB}.
\end{align}
\end{theorem}

\begin{proof}
The large-deviations upper bound in inequality \eqref{eq: UBB} may be established via a direct application of the G\"{a}rtner-Ellis theorem \cite{dembo1998large}, stated in Appendix \ref{appendix: gartnerellis}, as follows: We note that $\log \sum_j \mu_n(j) =  G_n(0) := \log s_n$. Define the probability measure $p_n := \frac{\mu_{n/n}}{s_n}.$ The log moment generating function of $p_n$ is given by
\begin{align*}
\cG_n(t) &= \log \sum_{j=0}^n p_n(j/n)e^{jt/n}
= G_n(t/n) - \log s_n.
\end{align*}
Thus, 
\begin{align*}
\lim_{n \to \infty} \frac{1}{n} \cG_n(nt) &= \lim_{n \to \infty} \frac{G_n(t)}{n} - \frac{\log s_n}{n}
= \Lambda(t) - \Lambda(0).
\end{align*}
This pointwise convergence (condition $(*)$ in Appendix \ref{appendix: gartnerellis}) is the key condition required to apply the G\"{a}rtner-Ellis theorem. We then obtain
\begin{align*}
\limsup_{n \to \infty} \frac{1}{n} \log p_n(I) &\leq -\inf_{x \in I} (\Lambda(x) - \Lambda(0))^* = -\inf_{x \in I} \Lambda^*(x) - \Lambda(0),
\end{align*}
which immediately gives 
\begin{align*}
\limsup_{n \to \infty} \frac{1}{n} \log \mu_{n/n}(I) \leq -\inf_{x \in I} \Lambda^*(x),
\end{align*}
proving inequality \eqref{eq: UBB}.

\medskip


We now establish the lower bound \eqref{eq: LBB}. We construct a sequence $\{\hat\mu_n\}$ such that $\mu_n \geq \hat \mu_n$, for all $n$; i.e., $\mu_n$ pointwise dominates $\hat \mu_n$, for all $n$. The large-deviations lower bound for $\{\hat\mu_n\}$ will serve as a large-deviations lower bound for $\{\mu_n\}$. Fix $a \geq 1$. We express every $n \geq 1$ as $n = qa+r$, where $r < a$, and define 
$$\hat \mu_n = \left({\mu_a}\right) ^{\star q} \star \mu_r,$$
where $\left({\mu_a}\right)^{\star q}$ is $\mu_a$ convolved $q$ times. Since $\{\mu_n\}$ is super-convolutive, this definition ensures that 
$$\mu_n \geq \hat\mu_n.$$
Define \mbox{$\hat G_n(t) := \log \sum_{j=0}^n \hat\mu_n(j)e^{jt}.$} Then
\begin{align}\label{eq: gnhat}
\lim_{n \to \infty} \frac{1}{n} \hat G_n(t) &= \lim_{n \to \infty} \frac{1}{n}( qG_a(t) + G_r(t) )\\
&\stackrel{(a)} = \lim_{q \to \infty} \frac{1}{aq}( qG_a(t) )\\
&= g_a(t),
\end{align}
where $(a)$ is true because $G_r(t)$ is bounded and thus inconsequential in the limit. Applying the G\"{a}rtner-Ellis theorem for $\{\hat \mu_n\}$, and noting that $g_a(t)$ is differentiable, we obtain
\begin{align}
&\liminf_{n \to \infty} \frac{1}{n} \log \hat\mu_{n/n}(F) \geq -\inf_{x \in F} g_a^*(x).
\end{align}

To show that $\inf_{x \in F} \left(g_a\right)^*(x)$ tends to $\inf_{x \in F} \left(\Lambda\right)^*(x)$ as $a \to \infty$, we establish the following lemma:

\begin{lemma}[Proof in Appendix \ref{proof: lemma: inf open}]\label{lemma: inf open}
Let $\{f_n\}$ be a sequence of continuous convex function on $[a,b]$
converging pointwise to $f$, 
where $|f(x)| < \infty$ for all $x \in [a,b]$.
Let $F \subseteq [a,b]$ be a relatively open set, i.e. $F$ is either
an open interval $(c,d)$ with $a < c < d < b$ or $(c, b]$ with 
$a < c < b$ or $[a, d)$ with $a < d < b$, or $[a, b]$. 
Then 
\[
\lim_n \{ \inf_{x \in F} f_n(x) \} = \inf_{x \in F} f(x)~.
\]
\end{lemma}

Taking the limit as $a\to \infty$ and using Lemma \ref{lemma: inf open}, which is applicable because of the properties of $g_n^*$ given in Remark \ref{remark: lambda star props}, we have
\begin{align*}
\liminf_{n \to \infty} \frac{1}{n} \log \mu_{n/n}(F) &\geq -\inf_{x \in F} \lim_{a \to \infty} g_a^*(x) \stackrel{(a)}= -\inf_{x \in F} \left(\Lambda\right)^*(x),
\end{align*}
where $(a)$ is a consequence of Lemma \ref{lemma: lambda star props}, and the fact that $F$ is an open set. This concludes the proof of Theorem \ref{thm: lambda Tn}.
\end{proof}

\subsection{Candidate function for $h_X$}\label{subsection: hx and lambda}
In Section \ref{subsection: superconv}, we developed the theory of convergence of proper super-convolutive sequences. In this section, we aim to apply this theory to the particular case of sequences of intrinsic volumes, which we denote by $\{\mu^\epsilon_n\}_{n \geq 1}$. Note that Lemma \ref{lemma: superconv mu} already provides the super-convolutive property of $\{\mu^\epsilon_n\}_{n \geq 1}$. Thus, we may directly apply Lemma \ref{lemma: lambda Tn} to conclude the existence of a limit function $\Lambda^\epsilon(t)$ satisfying
\begin{equation*}
\Lambda^\epsilon(t) = \lim_{n \to \infty} \frac{G^\epsilon_n(t)}{n},
\end{equation*}
where 
\begin{align*}
G^\epsilon_n(t) = \log \left( \sum_{j=0}^n \mu^\epsilon_n(j)e^{jt} \right).
\end{align*}

Our first lemma in this section shows that $\{\mu^\epsilon_n\}_{n \geq 1}$ is a \emph{proper} super-convolutive sequence:

\begin{lemma}[Proof in Appendix \ref{proof: lemma: alpha beta gamma}]\label{lemma: alpha beta gamma}
Let $\epsilon > 0$, and let $\{\mu^\epsilon_n\}_{n \geq 1}$ be the sequence of intrinsic volumes of the one-sided typical sets of a log-concave random variable. Then this sequence is a proper super-convolutive sequence; i.e., it satisfies the following three conditions:
\begin{enumerate}
\item[(i)]
For all $n \geq 1$, we have $\mu^\epsilon_n(0) > 0$ and $\mu^\epsilon_n(n) > 0$.

\item[(ii)] The limit $\beta := \lim_n \frac{\log \mu^\epsilon_n(0)}{n}$ is finite.

\item[(iii)] The value of $\gamma := \Lambda^\epsilon(0) < \infty$.

\end{enumerate}
\end{lemma}
\begin{remark}\label{remark: alpha beta gamma}
Properties (i) and (ii) are immediate, but showing (iii) is non-trivial. Our proof proceeds by constructing a super-convolutive sequence $\{\mu^{\text{cp}}_n\}_{n \geq 1}$ that pointwise dominates $\{\mu^\epsilon_n\}_{n \geq 1}$. This sequence corresponds to the intrinsic volumes of a family of crosspolytopes $\{\cC_n\}_{n \geq 1 }$, which are constructed in such a way that $\cT^\epsilon_n \subseteq \cC_n$ for all $n \geq 1$. The limit $\Lambda^\text{cp}(0)$ may be bounded using the explicit formulae for intrinsic volumes of crosspolytopes, and this bound also serves as a bound for $\Lambda^\epsilon(0)$.
\end{remark}

Applying Theorem \ref{thm: lambda Tn} to the proper super-convolutive sequence $\{\mu^\epsilon_n\}_{n \geq 1}$, we conclude that the convex conjugate of $\Lambda^\epsilon$, denoted by $(\Lambda^\epsilon)^*$, characterizes the large-deviations type convergence of $\{\mu^\epsilon_n\}_{n \geq 1}$:

\begin{theorem}\label{thm: lambda Tn new}
Let $\{\mu^\epsilon_n\}$ be as in Lemma \ref{lemma: alpha beta gamma}. Define a sequence of measures $\{\mu^\epsilon_{n/n}\}_{n\geq 1}$ supported on $[0,1]$ by assigning point masses as follows:
$$\mu^\epsilon_{n/n}\left(\left\{\frac{j}{n}\right\}\right) := \mu^\epsilon_n(j),  \text{~~for~~}  0 \leq j \leq n.$$
 Let $I \subseteq \mathbb{R}$ be a closed set and $F \subseteq \mathbb{R}$ be an open set. Then 
\begin{align}
\limsup_{n \to \infty} \frac{1}{n} \log \mu^\epsilon_{n/n}(I) &\leq - \inf_{x \in I} (\Lambda^\epsilon)^*(x), \qquad \text{and}\label{eq: UBB new}\\
\liminf_{n \to \infty} \frac{1}{n} \log \mu^\epsilon_{n/n}(F) &\geq - \inf_{x \in F} (\Lambda^\epsilon)^*(x)\label{eq: LBB new}.
\end{align}
\end{theorem}
\begin{proof}		
The proof is immediate using Lemma \ref{lemma: alpha beta gamma} and Theorem \ref{thm: lambda Tn}.
\end{proof}

It is now natural to conjecture that 
$$h_X(\theta) = \lim_{\epsilon \to 0} -(\Lambda^\epsilon)^*(\theta).$$
To show that this limit does exist, we establish the following theorem:
\begin{theorem}\label{thm: patpatpat}
Define the function $-\Lambda^*:[0,1] \to \mathbb R$ as the pointwise limit of $-\left(\Lambda^{\epsilon}\right)^*$, as $\epsilon \to 0_+$:
\begin{equation}
-\Lambda^*(\theta) := \lim_{\epsilon \to 0_+} -\left(\Lambda^{\epsilon}\right)^*(\theta), \text{~~for~~} \theta\in[0,1].
\end{equation}
Then $-\Lambda^*$ is a continuous, concave function on $[0,1]$.
\end{theorem}
\begin{proof}
From the definition of a typical set \eqref{eq: def Tn 1}, it is easy to see that for $\epsilon_1 < \epsilon_2$, the corresponding typical sets satisfy $\cT^{\epsilon_1}_n \subseteq \cT^{\epsilon_2}_n$, for all $n \geq 1$. Using the monotonicity of intrinsic volumes with respect to inclusion, we have $-\left(\Lambda^{\epsilon_1}\right)^* \leq -\left(\Lambda^{\epsilon_2}\right)^*$. Thus, for each $\theta \in [0,1]$, the value of $-\left(\Lambda^{\epsilon}\right)^*(\theta)$ monotonically decreases as $\epsilon \to 0_+$. To ensure that the quantity does not tend to $-\infty$ and establish a pointwise convergence result, we first provide a lower bound. Fix an $\epsilon_0 > 0$. From Lemma \ref{lemma: lambda star props} and Remark \ref{remark: lambda star props} we have
\begin{align}
-\left(\Lambda^{\epsilon_0}\right)^*(0) &\geq \beta = 0, \qquad \text{and} \label{eq: star psi 1}\\
-\left(\Lambda^{\epsilon_0}\right)^*(1) &\geq \alpha = h(X) - \epsilon_0. \label{eq: star psi 2}
\end{align}
By the concavity of $-\left(\Lambda^{\epsilon}\right)^*$, we obtain the linear lower bound  $-\left(\Lambda^{\epsilon}\right)^*(\theta) \geq\theta(h(X)-\epsilon)$, for all $\theta \in [0,1]$. Thus, as $\epsilon \to 0_+$, the value of $-\left(\Lambda^{\epsilon}\right)^*$ may be uniformly lower-bounded. We then use the following lemma to conclude the proof:
\begin{lemma}[Proof in Appendix \ref{proof: lemma: lambda without epsilon}]\label{lemma: lambda without epsilon}
Let $\{f_n\}$ be a sequence of continuous, concave functions on $[a, b]$, converging pointwise and in a monotonically decreasing manner to a function $f$. Then $f$ is a continuous, concave function on $[a, b]$.
\end{lemma}
\end{proof}
 The function $-(\Lambda^*)$ obtained in this fashion is the candidate function we have sought. In the subsequent sections, we show that $-(\Lambda^*)(\theta) = h_X(\theta)$, for all $\theta \in [0,1]$.


\section{Proof of the main theorem: Part A}\label{section: part a}

\begin{theorem}[Part A of Theorem \ref{thm: main}]\label{thm: part a}
Let $-\Lambda^*$ be as in Theorem \ref{thm: patpatpat}. The following equality holds:
\begin{equation}
\lim_{\epsilon \to 0_+} \lim_{n \to \infty} \frac{1}{n} \log \mu^\epsilon_n({\lfloor n\theta \rfloor}) = -\Lambda^*(\theta), \quad \text{ for all } \theta \in (0,1).
\end{equation}
\end{theorem}

The key result that bridges the gap between the large-deviations type convergence in Theorem \ref{thm: lambda Tn new} and Theorem \ref{thm: part a} is the Alexandrov-Fenchel inequality \cite{schneider2013convex}. While the complete inequality is quite general, holding not just for intrinsic volumes, but also for mixed volumes, we state here a version from McMullen \cite{mcmullen1991inequalities}:

\begin{theorem}[Inequality for intrinsic volumes \cite{mcmullen1991inequalities}]\label{thm: mcmullen}
Let $X \subseteq \mathbb R^n$ be a compact, convex set, and let its intrinsic volumes be $\{\mu_n(0), \dots, \mu_n(n)\}$. Then the following inequality holds:
\begin{equation}
\mu_n(j)^2 \geq \frac{j+1}{j}\mu_n(j-1)\mu_n(j+1), \quad \text{ for all } j \geq 1.
\end{equation}
\end{theorem} 

A immediate corollary of this theorem is the following:
\begin{corollary}\label{cor: mcmullen}
The sequence of intrinsic volumes is log-concave; i.e.,
\begin{equation}
\log \mu_n(j) \geq \frac{\log \mu_n(j-1) + \log \mu_n(j+1)}{2}, \quad \text{ for } 1 \leq j \leq n.
\end{equation}
\end{corollary}

\begin{proof}[Proof of Theorem \ref{thm: part a}]
Corollary \ref{cor: mcmullen} implies the log-concavity of the sequence $\{\mu^\epsilon_n(0), \dots, \mu^\epsilon_n(n)\}$. Our goal is to show that a family of log-concave sequences such as $\{\mu^\epsilon_n(\cdot)\}$, which converges in the large-deviation sense to $-\left(\Lambda^\epsilon\right)^*$, also converges pointwise to $-\left(\Lambda^\epsilon\right)^*$ in the following sense: For $\theta \in (0,1)$,
\begin{equation}\label{eq: ptwise mu}
\lim_{n\to \infty}\frac{1}{n} \log \mu^\epsilon_n(\lfloor n\theta \rfloor) = -\left(\Lambda^{\epsilon}\right)^*(\theta).
\end{equation}
For all $n \geq 1$, define the functions $a^\epsilon_n(\theta)$ by linearly interpolating the values of $a^\epsilon_n(j/n)$, where the value of $a^\epsilon_n(j/n)$ is given by
\begin{equation*}
a^\epsilon_n\left(\frac{j}{n}\right) = \frac{1}{n} \log \mu^\epsilon_n(j),  \text{~~for~~} 0\leq j \leq n.
\end{equation*}
The following lemma is a simple restatement of Corollary \ref{cor: mcmullen}:

\begin{lemma}[Proof in Appendix \ref{proof: lemma: alex}]\label{lemma: alex}
For each $n$ and $\epsilon > 0$, the function $a^\epsilon_n(\cdot)$ is concave.
\end{lemma}
Next, we show that the sequence of functions $\{a^\epsilon_n\}_{n \geq 1}$ converges pointwise---indeed, uniformly---on closed intervals:

\begin{lemma}[Proof in Appendix \ref{proof: lemma: sn linearized}]\label{lemma: sn linearized}
Let $\cI \subset (0,1)$ be a closed interval. The sequence of functions $\{a^\epsilon_n\}$ converges uniformly to $-(\Lambda^\epsilon)^*(\theta)$ on $\cI$.
\end{lemma}
\begin{proof}[Proof sketch]
Although not difficult, the proof of this lemma is fairly technical. For $\theta \in \cI$, it is possible to upper-bound 
$\lim\sup_n a^\epsilon_n$ by a number arbitrarily close to $-\Lambda^*(\theta_0)$, by a straightforward application of Theorem \ref{thm: lambda Tn new} to a small interval $I_0$ that contains $\theta_0$. To lower-bound $\lim\inf_{n} a^\epsilon_n(\theta)$, we trap the interval $I_0$ between the intervals to its left and right, denoted by $I_{-1}$ and $I_{+1}$. We show that there exist $\theta_{-1} \in I_{-1}$ and $\theta_{+1} \in I_{+1}$ such that $a^\epsilon_n$ evaluated at both these points is close to $-\Lambda^*(\theta_0)$, and apply concavity provided by Corollary \ref{cor: mcmullen} to produce a lower bound for $\lim\inf_{n} a^\epsilon_n(\theta)$ that is also close to $-\Lambda^*(\theta_0)$. Taking the limit, we conclude the proof of Lemma \ref{lemma: sn linearized}.
\end{proof}

To complete the proof of Theorem \ref{thm: part a}, we take the limit as $\epsilon \to 0_+$ and use Theorem \ref{thm: patpatpat}, which produces the claimed result:

\begin{align*}
\lim_{\epsilon \to 0_+} \lim_{n \to \infty} \frac{1}{n} \log \mu^\epsilon_n({\lfloor n\theta \rfloor}) &= \lim_{\epsilon \to 0_+} \lim_{ n \to \infty} a^\epsilon_n\left(\frac{\lfloor n\theta \rfloor}{n} \right) =  \lim_{\epsilon \to 0_+} \lim_{ n \to \infty} a^\epsilon_n(\theta)\\
&=  \lim_{\epsilon \to 0_+} -(\Lambda^\epsilon)^*(\theta) = -\Lambda^*(\theta).
\end{align*}
\end{proof}

The above strategy does not succeed in establishing the theorem at the endpoints $\theta = 0$ and $\theta = 1$, however. The main difficulty is that the sandwiching argument described in the proof sketch of Lemma \ref{lemma: sn linearized} can no longer work, since there is no interval to sandwich the endpoints between. To settle these cases, we exploit some additional geometric properties of typical sets in the following two sections.


\section{Proof of the main theorem: Part B}\label{section: part b}

\begin{theorem}[Part B of Theorem \ref{thm: main}]\label{thm: part b}
Let $-\Lambda^*$ be as in Theorem \ref{thm: patpatpat}. The following equality holds:
\begin{equation}
\lim_{\epsilon \to 0_+} \lim_{n \to \infty} \frac{1}{n} \log \mu^\epsilon_n(0) = -\Lambda^*(0).
\end{equation}
This is equivalent to the claim $-\Lambda^*(0) = 0.$
\end{theorem}
\begin{proof}		
The proof of this result resembles the proof of Lemma \ref{lemma: alpha beta gamma}, particularly regarding the approach to proving $\gamma < \infty$ as (cf.\ Remark \ref{remark: alpha beta gamma}). In Lemma \ref{lemma: alpha beta gamma}, we constructed a sequence of crosspolytopes $\{\cC_n\}$ such that $\cT^\epsilon_n \subseteq \cC_n$ for all $n \geq 1$. We briefly describe how to achieve this in the present situation. Observe that for each log-concave distribution $p_X(x) = e^{-\Phi(x)}$, we may find constants $c_1 > 0$ and $c_2$ such that
\begin{align}
\Phi(x) \geq c_1|x| + c_2, \text{~~for all~~} x \in \mathbb R. \label{eq: psi linear 0}
\end{align}
For $A =  \frac{h(X)+\epsilon - c_2}{c_1}$, define the sequence of regular crosspolytopes $\{\cC_n\}_{n=1}^\infty$ defined by
\begin{align*}
\cC_n := \left\{x^n \in \mathbb R^n ~|~ \sum_{i=1}^n |x_i| \leq An \right\}.
\end{align*}
If $(x_1, x_2, \dots, x_n) \in \cT^\epsilon_n$, then by Definition \ref{def: typical},
\begin{align*}
n(h(X) + \epsilon) &\geq \sum_{i=1}^n \Phi (x_i) \geq c_1\sum_{i=1}^n |x_i| + nc_2. 
\end{align*}
Thus,
\begin{align*}
\sum_{i=1}^n |x_i| \leq  \frac{n(h(X)+\epsilon - c_2)}{c_1}= nA,
\end{align*}
giving us $(x_1, \dots, x_n) \in \cC_n$, and consequently, $\cT^\epsilon_n \subseteq \cC_n$.

\begin{lemma}\label{lemma: cross}
Let the intrinsic volumes of $\{\cC_n\}_{n\geq 1}$ be given by $\{\mu^{\text{cp}}_n\}_{n \geq 1}$. Then the following statements are true:
\begin{enumerate}
\item[(i)]
The sequence $\{\cC_n\}_{n \geq 1}$ is super-multiplicative; i.e.
\begin{equation}\label{eq: superconv mu cp}
\cC_m \times \cC_n \subseteq \cC_{m+n}, \quad \text{ for all   } \quad m,n \geq 1.
\end{equation}

\item[(ii)]
The sequence of intrinsic volumes $\{\mu^{\text{cp}}_n\}_{n \geq 1}$ is a proper super-convolutive sequence.

\item[(iii)] There exists a continuous and concave function $-(\Lambda^{\text{cp}})^*$ such that for all $\theta \in [0,1)$, the following equality holds:

\begin{equation}
\lim_{n \to \infty} \frac{1}{n} \log \mu^{\text{cp}}_n(\lfloor n\theta \rfloor) = -(\Lambda^{\text{cp}})^*(\theta).
\end{equation}
In particular, $-(\Lambda^{\text{cp}})^*(0) = 0$.
\end{enumerate}
\end{lemma}

\begin{proof}[Proof of Lemma \ref{lemma: cross}]
Part (i) is easily verified by following the same steps as in Lemma \ref{lemma: superconv Tn}. The proof of part (ii) is already contained in Lemma \ref{lemma: alpha beta gamma}, since we have shown that
$$\lim_{n \to \infty} \frac{1}{n} \log \left( \sum_{j=0}^n \mu^{\text{cp}}_n(j) \right) < \infty$$
via the explicit formulae for $\mu^{\text{cp}}_n$. Since $\{\mu^{\text{cp}}_n\}_{n \geq 1}$ is a proper super-convolutive sequence, we use Theorem \ref{thm: lambda Tn} to infer the existence of $-(\Lambda^{\text{cp}})^*$, establishing the convergence of $\{\mu^{\text{cp}}_n\}_{n\geq 1}$ in the large-deviations sense. Following exactly the same steps as in Theorem \ref{thm: part a}, we conclude the pointwise convergence 
\begin{equation}
\lim_{n \to \infty} \frac{1}{n} \log \mu^{\text{cp}}_n(\lfloor n\theta \rfloor) = -(\Lambda^{\text{cp}})^*(\theta), \quad \text{ for } \theta \in (0,1).
\end{equation}
Thus, the only part we need to show is $-(\Lambda^{\text{cp}})^*(0) = 0$. As in Lemma \ref{lemma: alpha beta gamma}, we exploit the fact the intrinsic volumes of crosspolytopes are available in closed form. By the continuity of $-(\Lambda^{\text{cp}})^*$, we have
\begin{align*}
-(\Lambda^{\text{cp}})^*(0) &= \lim_{\theta \to 0} -(\Lambda^{\text{cp}})^*(\theta) =  \lim_{\theta \to 0} \left[ \lim_{n \to \infty} \frac{1}{n} \log \mu_n^{\text{cp}}(\lfloor n\theta\rfloor)\right]. 
\end{align*}
The value of $\mu_n^{\text{cp}}(\lfloor n \theta \rfloor)$ is given by
\begin{equation*}
2^{\lfloor n \theta \rfloor+1}{n \choose \lfloor n \theta \rfloor+1}\frac{\sqrt{\lfloor n \theta \rfloor+1}}{\lfloor n \theta \rfloor!}\frac{(nA)^{\lfloor n \theta \rfloor}}{\sqrt \pi} \times \int_{0}^\infty e^{-x^2}\left(\frac{2}{\sqrt\pi} \int_{0}^{x/\sqrt{\lfloor n \theta \rfloor+1}} e^{-y^2} dy\right)^{n-\lfloor n \theta \rfloor-1}dx.
\end{equation*}
As shown in the proof of Lemma \ref{lemma: alpha beta gamma}, 
\begin{equation*}
\int_{0}^\infty e^{-x^2}\left(\frac{2}{\sqrt\pi} \int_{0}^{x/\sqrt{\lfloor n \theta \rfloor+1}} e^{-y^2} dy\right)^{n-\lfloor n \theta \rfloor-1}dx \leq \frac{\sqrt \pi}{2}.
\end{equation*}
Thus, we obtain
\begin{equation*}
-(\Lambda^{\text{cp}})^*(0) \leq \lim_{\theta \to 0} \lim_{n \to \infty}\frac{1}{n} \log \left( 2^{\lfloor n \theta \rfloor+1}{n \choose \lfloor n \theta \rfloor+1}\frac{\sqrt{\lfloor n \theta \rfloor+1}}{\lfloor n \theta \rfloor!}\frac{(nA)^{\lfloor n \theta \rfloor}}{\sqrt \pi} \times \frac{\sqrt \pi}{2}\right).
\end{equation*}
It is easy to verify that the right-hand side evaluates to 0, so
$$-(\Lambda^{\text{cp}})^*(0) \leq 0.$$
However, by Lemma \ref{lemma: lambda star props}, we also have 
$$-(\Lambda^{\text{cp}})^*(0) \geq \lim_{n \to \infty} \frac{1}{n} \log \mu^{\text{cp}}_n(0) = 0.$$  This shows that $-(\Lambda^{\text{cp}})^*(0) = 0$ and completes the proof.
\end{proof}

The containment $\cT^\epsilon_n \subseteq \cC_n$ implies 
\begin{equation*}
-(\Lambda^\epsilon)^*(\theta) \leq -(\Lambda^{\text{cp}})^*(\theta), \quad \text{ for all } \theta \in [0,1].
\end{equation*}
Thus, we obtain the inequality
\begin{equation*}
-(\Lambda^\epsilon)*(0) \leq -(\Lambda^{\text{cp}})^*(0) = 0.
\end{equation*}
Furthermore, Lemma \ref{lemma: lambda star props} implies that
$$-(\Lambda^\epsilon)^*(0) \geq 0.$$
This forces $-(\Lambda^\epsilon)^*(0) = 0$. Taking the limit as $\epsilon \to 0_+$, we arrive at $-\Lambda^*(0) = 0$, completing the proof.
\end{proof}

\section{Proof of the main theorem: Part C}\label{section: part c}

\begin{theorem}[Part C of Theorem \ref{thm: main}]\label{thm: part c}
Let $-\Lambda^*$ be as in Theorem \ref{thm: patpatpat}. The following equality holds:
\begin{equation}
\lim_{\epsilon \to 0_+} \lim_{n \to \infty} \frac{1}{n} \log \mu^\epsilon_n(n) = -\Lambda^*(1).
\end{equation}
This is equivalent to the claim $-\Lambda^*(1) = h(X).$
\end{theorem}
\begin{proof}
To prove this part, we use the following inequality for intrinsic volumes proved in Campi \& Gronchi \cite{campi2011estimates}:
\begin{theorem}[Loomis-Whitney type inequality for intrinsic volumes \cite{campi2011estimates}]\label{thm: campi}
Let $X \subseteq \mathbb R^n$ be a compact  set and let its intrinsic volumes be $\{V_0(X), \dots, V_n(X)\}$. Let $\{e_1, \dots, e_n\}$ be the standard basis for $\mathbb R^n$. For a set $S \subset \{1, 2, \dots, n\}$, denote $e_S = \{e_j | j \in S\}$. Denote $\{1, 2, \dots, n\} \setminus \{i\}$ by $\{i\}^c$. Let $X | S$ be the set obtained by orthogonally projecting $X$ on the space spanned by $e_S$.\\

\noindent For $0 \leq m \leq n-1$, the following inequality holds:
\begin{equation}\label{eq: Vm inequality}
V_m(X) \leq \frac{1}{n-m} \sum_{j=1}^n V_m(X \mid \{j\}^c),
\end{equation}
provided the intrinsic volumes of $X | \{j\}^c$ satisfy the stability condition 
\begin{equation}\label{eq: Vm condition}
(*) :V_m( X \mid \{j\}^c) \leq \frac{1}{n-m} \sum_{i=1}^n V_m(X \mid \{i\}^c), ~\forall j \le n.
\end{equation}
\end{theorem}

Our first lemma concerns intrinsic volumes of projections of typical sets:
\begin{lemma}\label{lemma: project Tn}
Let $S_1, S_2 \subset \{1, 2, \dots, n\}$ be such that $|S_1| = |S_2| = k$. Then the intrinsic volumes of $(\cT^\epsilon_n\mid S_1)$ are the same as those of $(\cT^\epsilon_n\mid S_2)$. In other words, intrinsic volumes of such projections depend only on the dimension and not the specific choice among the ${n \choose k}$ possible subspaces.
\end{lemma}
\begin{proof}		
Without loss of generality, let $S_1 = \{1, 2, \dots, k\}$ and $S_2 = \{u_1, u_2, \dots, u_k\}$. The projection $(\cT^\epsilon_n\mid S_1)$ is the set
\begin{equation*}
(\cT^\epsilon_n\mid S_1) = \left\{(x_1, \dots, x_k) \mid \exists (x_1, \dots, x_n) \text{ such that } \sum_{i=1}^n \phi(x_i) \leq n(h(X)+\epsilon)\right\}.
\end{equation*}
Suppose $\min_x \Phi(x) = \eta$. Then the above expression is equivalent to
\begin{equation}\label{eq: project eta}
(\cT^\epsilon_n\mid S_1) = \left\{(x_1, \dots, x_k) \mid \sum_{i=1}^k \phi(x_i) \leq n(h(X)+\epsilon) - (n-k)\eta\right\}.
\end{equation}
Similarly,
\begin{align*}
(\cT^\epsilon_n\mid S_2) &= \left\{(x_{u_1}, \dots, x_{u_k}) \mid \exists (x_1, \dots, x_n) \text{ such that } \sum_{i=1}^n \phi(x_i) \leq n(h(X)+\epsilon)\right\}\\
&= \left\{(x_{u_1}, \dots, x_{u_k}) \mid \sum_{i=1}^k \phi(x_{u_i}) \leq n(h(X)+\epsilon) - (n-k)\eta\right\}.
\end{align*}
It is now clear that these two sets are exactly the same, except in different dimensions. In particular, they have identical intrinsic volumes.
\end{proof}

Our next lemma shows that Theorem \ref{thm: campi} is applicable to sets of the form $(\cT^\epsilon_n\mid S)$:
\begin{lemma}\label{lemma: iterate}
Let $S = \{u_1, u_2, \dots, u_k\}  \subseteq \{1, 2, \dots, n\}$. Let $(\cT^\epsilon_n \mid S)$ be the set $\cT_S$. Let $0 \leq m \leq k-1$. Then $\cT_S$ satisfies the following inequality:
\begin{equation*}
V_m(\cT_S) \leq \frac{1}{k-m} \sum_{j=1}^k V_m(\cT_S \mid S \setminus \{u_j\}).
\end{equation*}
\end{lemma}
\begin{proof}
We only need to check that the condition $(*)$ holds for $\cT_S$; i.e., we need to check, for all $i$, that
\begin{equation}\label{eq: k-1}
V_m(\cT_S \mid S \setminus \{u_i\}) \leq \frac{1}{k-m} \sum_{j=1}^k V_m(\cT_S \mid S \setminus \{u_j\}).
\end{equation}
The sets $(\cT_S \mid S \setminus \{u_i\})$ are simply $(k-1)$-dimensional projections of $\cT^\epsilon_n$, and by Lemma \ref{lemma: project Tn} all such sets have identical intrinsic volumes. This trivially yields inequality \eqref{eq: k-1} and completes the proof.
\end{proof}

With Lemma \ref{lemma: iterate} established, we are now in a position to iterate the inequality \eqref{eq: Vm inequality} in Theorem \ref{thm: campi} until we are left with $m$-dimensional projections of $\cT^\epsilon_n$. Note that each of the ${n \choose m}$ orthogonal subspaces is counted $(n-m)!$ times in the iteration. Additionally, there is a factor of $\frac{1}{(n-m)(n-m-1)\cdots 1} = \frac{1}{(n-m)!}$ outside the summation. These factors cancel out, yielding the key inequality

\begin{align}
V_m(\cT^\epsilon_n) &\leq \!\!\! \sum_{1 \leq i_1< \dots< i_{m} \leq n} V_m( \cT^\epsilon_n \mid \{{i_1}, \dots, {i_{m}}\}) \notag \\
 &= {n \choose m}V_m(\cT^\epsilon_n \mid \{1, 2, \dots, m\}) \notag \\
 &= {n \choose m} |(\cT^\epsilon_n \mid \{1, 2, \dots, m\})|. \label{eq: VmVm}
\end{align}

Our final lemma produces a bound on $|(\cT^\epsilon_n \mid \{1, 2, \dots, m\})|$:
\begin{lemma}\label{lemma: bound on vm}
Let $\eta = \min_x \Phi(x)$. Then we have the bound
\begin{equation}\label{eq: vol M}
|(\cT^\epsilon_n \mid \{1, 2, \dots, m\})| \leq e^{n(h(X) + \epsilon) - (n-m)\eta}.
\end{equation}
\end{lemma}
\begin{proof}	
Using equation \eqref{eq: project eta}, we have
\begin{equation*}
(\cT^\epsilon_n \mid \{1, 2, \dots, m\}) = \left\{(x_1, \dots, x_m) \mid \sum_{i=1}^m \Phi(x_i) \leq n(h(X) + \epsilon) - (n-m)\eta \right\},
\end{equation*}
where $\eta = \min_x \Phi(x)$. Rewriting, we have
\begin{equation*}
(\cT^\epsilon_n \mid \{1, 2, \dots, m\}) = \left\{(x_1, \dots, x_m) \mid \prod_{i=1}^m p_X(x_m) \geq e^{-n(h(X) + \epsilon) + (n-m)\eta} \right\}.
\end{equation*}
Since the probability density over $(\cT^\epsilon_n \mid \{1, 2, \dots, m\})$ is at most 1, we obtain the upper bound
\begin{equation*}
|(\cT^\epsilon_n \mid \{1, 2, \dots, m\})| \leq e^{n(h(X) + \epsilon) - (n-m)\eta}, 
\end{equation*}
completing the proof.
\end{proof}

For $\theta \in (0,1)$, choose $m = \lfloor n\theta \rfloor$. Substituting in inequality \eqref{eq: VmVm} and using inequality \eqref{eq: vol M}, we obtain 
\begin{equation*}
\mu^\epsilon_n(\lfloor n\theta \rfloor)  \leq {n \choose \lfloor n\theta \rfloor}e^{n(h(X)+\epsilon) - (n-\lfloor n\theta \rfloor)\eta}.
\end{equation*}
Taking logarithms of both sides and dividing by $n$, we have
\begin{equation*}
\frac{1}{n} \log \mu^\epsilon_n(\lfloor n\theta \rfloor) \leq \frac{1}{n}\log \left({n \choose \lfloor n\theta \rfloor}e^{n(h(X)+\epsilon) - (n-\lfloor n\theta \rfloor)\eta} \right).
\end{equation*}
Taking the limit as $n \to \infty$ and using Theorem \ref{thm: part b}, we then have
\begin{equation*}
-\left(\Lambda^{\epsilon}\right)^*(\theta) \leq H(\theta) + h(X) + \epsilon - (1-\theta)\eta.
\end{equation*}
Taking the limit as $\epsilon \to 0_+$ and using Theorem \ref{thm: patpatpat}, we obtain
\begin{equation*}
-\Lambda^*(\theta) \leq H(\theta) + h(X) - (1-\theta) \eta.
\end{equation*}
Taking the limit as $\theta \to 1$ and using the continuity of $-\Lambda^*$ from Theorem \ref{thm: patpatpat}, we also have 
$$-\Lambda^*(1) \leq h(X).$$
Additionally, we have the lower bound from Lemma \ref{lemma: lambda star props}, which asserts that 
$$-\Lambda^*(1) \geq \lim_{\epsilon \to 0_+} \lim_{n \to \infty} \frac{1}{n} \log \mu^\epsilon_n(\cT^\epsilon_n) = h(X).$$ 
This forces $-\Lambda^*(1) = h(X)$ and completes the proof.
\end{proof}


\section{Conclusion}\label{section: discussion}

In this paper, we have established an entropic lifting of the concept of intrinsic volumes for the special case of log-concave distributions. For a log-concave random variable $X$, the function $h_X(\cdot)$ may be interpreted as a generalization of the entropy functional, in a similar way as intrinsic volumes are considered a generalization of the volume functional. We now briefly describe several future research directions and open problems.

The first natural question is whether it is possible to define intrinsic entropies for random variables that do not possess a log-concave density. Note that the main reason we focused on log-concave random distributions was that it is very straightforward to define \emph{convex} typical sets for these distributions, and convex sets have well-defined intrinsic volumes. However, the concept of intrinsic volumes is not just restricted to convex sets or polyconvex sets, and may be extended to a larger class of sets. Federer \cite{Fed59} defined the concept of sets with \emph{positive reach}, as follows: The reach of a set $A \subset \mathbb R^n$ is the largest value $r$ such that for all $x$ satisfying $d(x, A) < r$, the set $A$ contains a unique point that is nearest to $x$. Here, $d(x,A)$ is the minimum of the distance between $x$ and some point in $A$. A set has a positive reach if $r>0$. Federer  \cite{Fed59} showed that it is possible to define intrinsic volumes for all sets having a positive reach. One possible approach to defining intrinsic entropies for distributions that are not log-concave would be to show that typical sets have a positive reach, and then take the limits of the intrinsic volumes as defined in  \cite{Fed59}. A paper by Schanuel \cite{Sch86} neatly illustrates where the theory of convex sets breaks down when we consider non-convex sets and suggests ways to surmount such difficulties. 

Another natural question is the following: How crucially does the value of $h_X(\theta)$ depend on the specific way we have defined typical sets in Definition \ref{def: typical}? Suppose we had an alternate definition for a sequence of (convex) sets $\{ \tilde \cT^\epsilon_n \}$, satisfying the following two properties:
\begin{enumerate}
\item[(i)]
$\lim_{\epsilon \to 0} \lim_{n \to \infty} \frac{1}{n} \log |\tilde \cT^\epsilon_n| = h(X)$, and
\item[(ii)]
$\lim_{n \to \infty} P(\tilde \cT^\epsilon_n) = 1$.
\end{enumerate}
It is interesting to ask if this sequence would exhibit the same limit for its intrinsic volumes as $\{\cT^\epsilon_n\}$. One example of such an alternate sequence is as follows: Suppose $X$ satisfies $EX = 0$ and $\text{Var}(X) = \sigma^2$. We define
$$\tilde \cT^\epsilon_n =  \cT^\epsilon_n \cap B_n(\sigma\sqrt{n} + \epsilon).$$ 
For such alternate definitions, the super-multiplicative property of typical sets in expression \eqref{eq: superconv Tn} may no longer hold. Thus, we may not be able to use the same analysis as presented in this paper to evaluate the limit of the intrinsic volumes. We believe the limit remains $h_X(\cdot)$ for any alternate definition of convex typical sets satisfying properties (i) and (ii), so it is a intrinsic property of the distribution itself. Although this problem is still open, we show in Appendix \ref{appendix: alternate} that  for many natural definitions of typical sets, intrinsic volumes grow at the same shared rate $h_X$.
%

In this paper, we have considered one-dimensional random variables. However, it is not hard to see that similar results also hold for multi-dimensional random variables. Thus, it is possible to generalize concepts such as joint entropy, conditional entropy, and mutual information by replacing entropy by intrinsic entropy in the definitions. The main focus of this paper was showing the existence of $h_\theta(X)$, as opposed to studying the particular properties of $h_X$. An important first step towards this would be to check whether intrinsic entropies satisfy an analog of the Brunn-Minkowski inequality. Here, we conjecture a version of the entropy power inequality, inspired by the complete Brunn-Minkowski inequality for intrinsic volumes \cite{schneider2013convex}: 
\begin{equation}\label{eq: h theta epi}
e^{\frac{2h_\theta(X)}{\theta}}+e^{\frac{2h_\theta(Y)}{\theta}} \leq e^{\frac{2h_\theta(X+Y)}{\theta}},
\end{equation}
where we recover the usual EPI for $\theta = 1$.
\section*{Acknowledgements}
The research of the authors was supported by NSF grant ECCS-1343398 and the NSF Science \& Technology Center grant CCF-0939370, Science of Information. VA also acknowledges support from the NSF grants CNS 1527846 and CCF 1618145.


\bibliographystyle{plain}	
\bibliography{myrefs}		
\begin{appendix}

\section{Proofs for Section \ref{section: candidate}}

\subsection{Proof of Lemma \ref{lemma: lambda props}}\label{proof: lemma: lambda props}

\begin{enumerate}
\item[(i)]
Condition \eqref{eq: supc condition on G} implies that
\begin{equation}
nG_1(t) \leq G_n(t),
\end{equation}
which implies
\begin{equation}
 g_1(t) \leq g_n(t).
\end{equation}
Taking the limit in $n$, we see that $g_1(t) \leq \Lambda(t)$.\\
 
For every $n$ and every $t \geq 0$,
\begin{align}
\frac{G_n(t)}{n}  &= \frac{1}{n} \log \sum_{j=0}^n \mu_n(j)e^{jt}\\
&\leq \frac{1}{n} \log \left( \left(\sum_{j=0}^n \mu_n(j)\right)e^{nt}\right)\\
&= \frac{G_n(0)}{n}  + t.
\end{align}
Taking the limit in $n$ and using $\gamma < \infty$, we see that for $\Lambda(t) \leq t+\gamma$ for $t \geq 0$. Similarly, for $t \leq 0$
\begin{align}
\frac{G_n(t)}{n}  &= \frac{1}{n} \log \sum_{j=0}^n \mu_n(j)e^{jt}\\
&\leq \frac{1}{n} \log \left(\sum_{j=0}^n \mu_n(j)\right)\\
&= \frac{G_n(0)}{n}.
\end{align}
Taking the limit in $n$ and using $\gamma < \infty$, we see that $\Lambda(t) \leq \gamma$ for $t \leq 0$.

\item[(ii)]
The functions $\{g_n\}$ are convex and monotonically increasing. Since $\Lambda$ is the pointwise limit of these functions, $\Lambda$ is also convex and monotonically increasing.
\end{enumerate}

\subsection{Proof of Lemma \ref{lemma: lambda star props}}\label{proof: lemma: lambda star props}
\begin{enumerate}
\item[(i)]
Note that the convex conjugates of the functions $g_1(t)$ and $\max(\gamma, t+\gamma)$ are both supported on $[0,1]$. By Lemma \ref{lemma: lambda props}, the function $\Lambda$ is trapped between these two functions, and thus $\Lambda^*$ is also supported on $[0,1]$. Furthermore, since $\Lambda^*$ is a convex conjugate function, it is convex and lower-semicontinuous on its domain. Since the domain is a closed interval, we obtain that $\Lambda^*$ must be continuous on $[0,1]$.

\item[(ii)]
We'll need to show that $g_n^*$ converges pointwise to $\Lambda^*$ on $(0,1)$. Fix an $x \in (0,1)$, and define 
$$\arg \max_t xt - g_n(t) := t_n.$$
Clearly, $g_n^*(x) = xt_n - g_n(t_n)$. Note that
\begin{align*}
g_n^*(x) &\geq xt - g_n(t) \Big|_{t = 0}\\
& = -g_n(0)\\
&\stackrel{(a)} \geq -\gamma,
\end{align*}
where $(a)$ is because $\gamma = \sup_n g_n(0)$.  If $t > \frac{\gamma - \log \mu_1(1)}{1-x}$, then we have
 \begin{align*}
 xt - g_n(t) &< xt - g_1(t)\\
 &\leq xt - (t + \log\mu_1(1))\\
 &= -(1-x)t - \log \mu_1(1) \\
 &< -(\gamma - \log \mu_1(1)) - \log \mu_1(1)\\
 &= -\gamma.
 \end{align*}
 This gives us that $t_n \leq \frac{\gamma - \log \mu_1(1)}{1-x}$. Similarly, if $t < \frac{\log \mu_1(0) -\gamma}{x}$, then
 \begin{align*}
 xt - g_n(t) &< xt - g_1(t)\\
 &\leq xt - \log \mu_1(0)\\
 &<(\log \mu_1(0) - \gamma) - \log \mu_1(0)\\
 &= -\gamma.
 \end{align*}
 This gives us that $t_n \geq \frac{\log \mu_1(0) -\gamma}{x}$. We can thus conclude that for all $n$,
\begin{equation}\label{eq: supc bounds on tn}
 t_n \in \left[\frac{\log \mu_1(0) -\gamma}{x}, \frac{\gamma - \log \mu_1(1)}{1-x}\right] := I_x.
\end{equation}

Note that all we used to prove relation \eqref{eq: supc bounds on tn} is that $g_n(t)$ is trapped between $g_1(t)$ and $\max(\gamma, t+\gamma)$. Since $\Lambda$ also satisfies this, we have
\begin{equation}\label{eq: supc lambda tn}
\arg\max_{t} xt - \Lambda(t) \in I_x.
\end{equation}
We now restrict our attention to the compact interval $I_x$. Let $\hat g_n$ be $g_n$ restricted to $I_x$. The convex functions $\hat g_n$ converge pointwise to a continuous limit $\hat \Lambda$, where $\hat \Lambda$ is $\Lambda$ restricted to $I_x$. This convergence must therefore be uniform \cite{jog2015geometric}, which implies convergence of $\hat g_n^*(x)$ to $\hat \Lambda^*(x)$. Furthermore, relation (\ref{eq: supc bounds on tn}) implies $\hat g_n^*(x)$ equals $g_n^*(x)$, and relation \eqref{eq: supc lambda tn} gives $\hat \Lambda^*(x)$ equals $\Lambda^*(x)$. Thus, $g_n^*(\cdot)$ converges pointwise to $\Lambda^*(\cdot)$ on $(0,1)$.\\

\item[(iii)]
To get the inequality for $t=0$ and $t=1$, we note that the function $\lim_{n} g_n^*(t)$ and $\Lambda^*$ are both convex on $[0,1]$ with $\Lambda^*$ also being continuous on $[0,1]$. As both these function agree on $(0,1)$, it is immediate that $\Lambda^*(0) \leq -\beta$ and $\Lambda^*(1) \leq  -\alpha$.
\end{enumerate}

\subsection{Example showing strict inequality at $0$}\label{example: strict inequality}
We show that it is possible for a proper super-convolutive sequences obey the strict inequality in Lemma \ref{lemma: lambda star props} (iii), using the following example:
\begin{example}\label{remark: alpha}
Let $\delta, \alpha >0$ be such that $\alpha > 1$ and $\delta < \frac{1}{2}$. For $n \geq 1$, consider the sequence $\{\mu_n(\cdot)\}$, given by
\begin{equation}
\label{eq: mu_n}
\mu_n(i) = 
\begin{cases}
{n-1 \choose i}\alpha^{i}, &\text{ for } 0 \leq i \leq n-1,\\
\delta, &\text{ for } i = n.\\
\end{cases}
\end{equation}
By an explicit calculation, we may show that for this sequence is a proper super-convolutive sequence, and we have the strict inequality $\Lambda^*(1) = -\log \alpha < 0$.
\end{example}


\subsection{Proof of Lemma \ref{lemma: inf open}}\label{proof: lemma: inf open}
We first prove the following lemma:
\begin{lemma}\label{lemma: venkat}
Let $\{g_n\}$ be a sequence of continuous convex function on $[c,d]$
converging pointwise to $g$, where $g$ is a continuous function. 
Let $F \subseteq [c,d]$ be a relatively open set.
Then 
\[
\lim_n \{ \inf_{x \in F} g_n(x) \} = \inf_{x \in F} g(x)~.
\]
\end{lemma}
\begin{proof}
By Lemma \ref{lemma: convex}, we know that $g_n$ converge uniformly to $g$ on
$[c,d]$. Thus, for any $\epsilon > 0$ we have 
\[
g_n(x) - \epsilon < g(x) < g_n(x) + \epsilon~,
\]
uniformly over $x \in [c, d]$, for all $n$ sufficiently large.
Let $x_n \in F$ be such that
\[
\inf_{x \in F} g_n(x) \le g_n(x_n) < \inf_{x \in F} g_n(x) + \epsilon~.
\]
Such an $x_n$ exists by the relative openness of $F$ and the 
continuity of $g_n$ on $[c,d]$.
Then, for all $n$ sufficiently large we have
\[
\inf_{x \in F} g(x) \le g(x_n) < g_n(x_n) + \epsilon 
< \inf_{x \in F} g_n(x) + 2 \epsilon~,
\]
and we also have
\[
\inf_{x \in F} g(x) \ge \inf_{x \in F} (g_n(x) - \epsilon)
= ( \inf_{x \in F} g_n(x) ) - \epsilon~,
\]
which completes the proof.
\end{proof}

The function $f$, being a pointwise limit of the convex functions 
$f_n$ on $[a,b]$,
is convex on $[a,b]$. 
Since $|f(x)| < \infty$ for all $x \in [a,b]$, we may conclude that
$f$ is continuous on $(a,b)$. Let $F \subseteq [a,b]$ be the relatively open
subset of interest. The convexity of $f$ on $[a,b]$ implies that 
\begin{equation}		\label{add.01}
\inf_{x \in F} f(x) = \lim_{\delta \to 0} 
\inf_{x \in F \cap [a+\delta, b-\delta]} f(x)~.
\end{equation}
Let $\delta > 0$ be sufficiently small that 
$F \cap [a+\delta, b-\delta] \neq \emptyset$. 
By Lemma \ref{lemma: convex} we conclude that, for
any $\delta > 0$, $f_n$ converges uniformly to $f$ on
$[a+\delta, b-\delta]$. By  Lemma \ref{lemma: venkat}, we conclude that
\[
\lim_n \{ \inf_{x \in F \cap [a+\delta, b-\delta]} f_n(x) \} 
= \inf_{x \in F \cap [a+\delta, b-\delta]} f(x)~.
\]
Since $\inf_{x \in F \cap [a+\delta, b-\delta]} f_n(x)$ 
is nondecreasing in $\delta$, we have
\[
\limsup_n \inf_{x \in F} f_n(x)
\le \inf_{x \in F \cap [a+\delta, b-\delta]} f(x)~,
\]
for all $\delta > 0$. Taking the limit as $\delta \to 0$ and using
(\ref{add.01}) we get
\[
\limsup_n \inf_{x \in F} f_n(x) \le \inf_{x \in F} f(x)~.
\]

For the inequality in the opposite direction, let $M < \infty$ be such
that $f_n(a) \le M$ and $f_n(b) \le M$ for all $n \ge 1$.
Such an $M$ exists under the assumptions of Lemma \ref{lemma: inf open}. 
For any $\delta > 0$ sufficiently small and any
$b - \delta < x \le b$, from the convexity of $f_n$, we have
\[
f_n(b-\delta) \le \frac{b-\delta -a}{x -a} f_n(x) 
+ \frac{x - b + \delta}{x - a} M~.
\]
This can be rearranged to read
\[
f_n(x) \ge \frac{x-a}{b - \delta -a} f_n(b-\delta) 
- \frac{x - b+ \delta}{b - \delta -a} M~.
\]
We can similarly get a lower bound on $f_n(x)$ for $a \le x < a + \delta$.
We conclude that there is a universal constant $C < \infty$ such that for all
$F \subset [a, b]$ that are relatively open, for 
$\delta > 0$ sufficiently small (depending on $F$), all $n \ge 1$, all
$x \in (b - \delta, b]$, and all $x \in [a, a + \delta)$, we have
\[
f_n(x) \ge \inf_{x \in F \cap [a + \delta, b - \delta]} f_n(x) - C \delta~.
\]
It follows that 
\[
\inf_{x \in F} f_n(x) 
\ge \inf_{x \in F \cap [a + \delta, b - \delta]} f_n(x) - C \delta~.
\]
Fix $\delta$ and let $n \to \infty$ and use (\ref{add.01}) to conclude
that
\[
\liminf_n \inf_{x \in F} f_n(x) 
\ge \inf_{x \in F \cap [a + \delta, b - \delta]} f(x) - C \delta~.
\]
Next, let $\delta \to 0$ and use the convexity of $f$ to conclude that
\[
\liminf_n \inf_{x \in F} f_n(x) 
\ge  \inf_{x \in F} f(x)~.
\]
This completes the proof of Lemma \ref{lemma: inf open}.


\subsection{Proof of Lemma \ref{lemma: alpha beta gamma}}\label{proof: lemma: alpha beta gamma}
Note that  $\mu^\epsilon_m \star \mu^\epsilon_n(m+n) = \mu^\epsilon_n(n)\mu^\epsilon_m(m)$ and $\mu^\epsilon_m \star \mu^\epsilon_n(0) = \mu^\epsilon_n(0)\mu^\epsilon_m(0)$. Thus the existence of the limits defining $\alpha$ and $\beta$ is given by sub-additivity. Existence of limit defining $\gamma$ follows from the equality $\gamma = \Lambda^{\epsilon}(0)$.
\begin{enumerate}
\item[(i) and (ii):]
The value of $\mu_n(0)$ is the Euler characteristic, which equals $1$ when $\cT_n^\epsilon$ is non-empty. We show that for every $n \geq 1$, the set $\cT^\epsilon_n$ has a nonempty interior; i.e., $\text{Vol}(\cT^\epsilon_n) = \mu_n(n) > 0$. Let $M = \max_x p_X(x) = e^{-\min_x \Phi(x)}$. Note that the set of minimizers of $\Phi$ is a nonempty set, since $\Phi \to +\infty$ as $|x| \to +\infty$. Let $x^*$ be any such minimizer of $\Phi$. For the point $(x^*, \dots, x^*) \in \mathbb R^n$, we have
\begin{align*}
 \sum_{i=1}^n \Phi (x_i) &= -n\log M.
\end{align*}
We also have the inequality
\begin{align*}
-h(X) &= \int_{\mathbb R}p_X(x)\log p_X(x) dx \leq \int_{\mathbb R}p_X(x)\log M dx 
= \log M.
\end{align*}
Thus, for the point $(x^*, \dots, x^*)$, we have
\begin{align*}
\sum_{i=1}^n \Phi (x_i) &= -n\log M < n(h(X)+\epsilon),
\end{align*}
so $(x^*, \dots, x^*) \in \cT_n^\epsilon$. By the continuity of $\Phi$ at $x^*$, we conclude that $\cT_n^\epsilon$ has a nonempty interior.

\item[(iii)]
Since $\Phi(x)  \to \pm\infty$ as $|x| \to \pm \infty$ and $\Phi$ is convex, we may find constants $c_1 > 0$ and $c_2$ such that
\begin{align}
\Phi(x) \geq c_1|x| + c_2, \text{~~for all~~} x \in \mathbb R. \label{eq: psi linear}
\end{align}
We start by showing that for $A :=  \frac{h(X)+\epsilon - c_2}{c_1}$, the sequence of regular crosspolytopes $\{\cC_n\}_{n=1}^\infty$ defined by
\begin{align*}
\cC_n := \{x^n \in \mathbb R^n ~|~ \sum_{i=1}^n |x_i| \leq An \}
\end{align*}
satisfies the containment 
\begin{equation}\label{eq: Tn in Cn}
\cT^\epsilon_n  \subseteq \cC_n, \text{~~for~~} n \geq 1. 
\end{equation}
For $x^n \in \cT^\epsilon_n$, using definition \eqref{eq: def Tn 4} and inequality \eqref{eq: psi linear}, we have 
\begin{align*}
\sum_{i=1}^n (c_1|x_i| + c_2)&\leq \sum_{i=1}^n \Phi (x_i) \leq  n(h(X)+\epsilon),
\end{align*}
implying that
\begin{align*}
 \sum_{i=1}^n |x_i| \leq n\left(\frac{h(X)+\epsilon - c_2}{c_1}\right),
\end{align*}
so $x^n \in \cC_n$. Hence, $\cT^\epsilon_n \subseteq \cC_n$, as claimed. Let the intrinsic volumes of $\cC_n$ be $ \mu^{\text{cp}}_n(\cdot)$. Note that $\mu_n(i) \leq  \mu^{\text{cp}}_n(i)$ for all $0 \leq i \leq n$, by the containment \eqref{eq: Tn in Cn}. Thus, $\gamma \leq \gamma^{\text{cp}}$, where 
\begin{equation}
\gamma^{\text{cp}} := \lim_{n \to \infty} \frac{1}{n} \log \left(\sum_{i=0}^n  \mu^{\text{cp}}_n(i)\right).
\end{equation}
We claim that $ \gamma^{\text{cp}} < \infty$. Define 
$$ G^{\text{cp}}_n(t) = \log \sum_{i=0}^n  \mu^{\text{cp}}_n(i) e^{it}, \text{~~and~~}  g^{\text{cp}}_n(t) = \frac{ G^{\text{cp}}_n(t)}{n}.$$
Note that the sequence $\{\cC_n\}$ is super-convolutive, so $ g^{\text{cp}}_n(t)$ converges pointwise. In particular,
for $t=0$, we have
\begin{align*}
 \gamma^{\text{cp}} = \lim_{n \to \infty} \frac{1}{n} \log \left(\sum_{i=0}^n \mu^{\text{cp}}_n(i) \right) \text{~~exists, and is possibly~~} +\infty.
\end{align*}
The $i$-th intrinsic volume of $\cC_n$ is given by \cite{betke1993intrinsic}
\begin{align*}
\mu^{\text{cp}}_n(i) = 
\begin{cases}
2^{i+1}{n \choose i+1}\frac{\sqrt{i+1}}{i!}\frac{(nA)^i}{\sqrt \pi} \times\int_{0}^\infty e^{-x^2}\left(\frac{2}{\sqrt\pi} \int_{0}^{x/\sqrt{i+1}} e^{-y^2} dy\right)^{n-i-1}dx \quad &\text{ if }  i \leq n-1\\
\frac{2^n}{n!}(nA)^n &\text{ if } i = n.
\end{cases}
\end{align*}
Note that
\begin{align*}
&\int_{0}^\infty e^{-x^2}\left(\frac{2}{\sqrt\pi} \int_{0}^{x/\sqrt{i+1}} e^{-y^2} dy\right)^{n-i-1}dx \\
&\leq \int_{0}^\infty e^{-x^2}\left(\frac{2}{\sqrt\pi} \int_{0}^{\infty} e^{-y^2} dy\right)^{n-i-1}dx\\
&\leq \int_{0}^\infty e^{-x^2} dx\\
&= \frac{\sqrt\pi}{2}.
\end{align*}
Thus, for $0 \leq i \leq n-1$,
\begin{align*}
\hat \mu_n(i) &\leq 2^{i+1}{n \choose i+1}\frac{\sqrt{i+1}}{i!}\frac{(nA)^i}{\sqrt \pi} \times \frac{\sqrt\pi}{2}\\
&=  2^{i}{n \choose i+1}\frac{\sqrt{i+1}}{i!}(nA)^i\\
&\leq 2^n \times 2^n \times \sqrt{n+1} \times A^i \times \frac{n^i}{i!}\\
&\leq 2^{2n}\sqrt{n+1} \times \max(1, A^n) \times  \frac{n^n}{n!}.\\
\end{align*}
We may check that the inequality also holds for $i=n$. Hence,
\begin{align*}
&\frac{1}{n} \log\left(\sum_{i=0}^n \mu^{\text{cp}}_n(i) \right)\\
 &\leq \frac{1}{n}\log \left((n+1)\times 2^{2n}\sqrt{n+1} \times \max(1, A^n) \times  \frac{n^n}{n!} \right).
\end{align*}
Taking the limit as $n \to \infty$, we obtain
\begin{align*}
\lim_{n \to \infty} \frac{1}{n} \log\left(\sum_{i=0}^n \mu^{\text{cp}}_n(i) \right) &\leq  2\log 2 + \max(0, \log A) + 1\\
& < \infty.
\end{align*}
This shows that $\gamma^{\text{cp}}$ is finite, and thus $\gamma$ is also finite.
\end{enumerate}


\subsection{Proof of Lemma \ref{lemma: lambda without epsilon}}\label{proof: lemma: lambda without epsilon}
Without loss of generality, take $a = 0$ and $b = 1$. Since $f$ is the pointwise limit of concave functions, it is also concave. The continuity of $f$ is not obvious a priori: it could be discontinuous at the endpoints 0 and 1.  Let $f(0) = \ell_0$ and $f(1) = \ell_1$. For any $n \geq 1$, the function $f_n$ is lower-bounded by the line joining $(0, \ell_0)$ and $(1, \ell_1)$. Call this lower bound $L(\theta)$, for $\theta \in [0,1]$.  We prove continuity at $0$, by showing that for $\eta > 0$, there exists a $\delta>0$ such that for $\theta \in [0, \delta)$, we have $|f(\theta)- \ell_0| < \eta$. Pick $N$ large enough such that
 $$f_N(0) - \ell_0 < \eta/2.$$ The function $f_N$ is continuous on $[0,1]$, so there exists $\delta_1 > 0$ such that for $\theta \in [0, \delta_1)$, we have
 $|f_N(\theta) - f_N(0)| < \eta/2$. Now pick a $\delta_2$ such that $|L(\theta) - \ell_0| < \eta/2$, for $\theta \in [0, \delta_2)$. Let $\delta = \min(\delta_1, \delta_2)$. For $n > N$, we have
 $$L(\theta) \leq f_n(\theta) \leq f_N(\theta).$$
 Thus, for $\theta \in [0, \delta)$, we obtain
 \begin{align*}
 f_n(\theta) \leq f_N(\theta) \leq f_N(0) + \eta/2 \leq \ell_0 + \eta,
 \end{align*}
 and
 \begin{align*}
 f_n(\theta) \geq L(\theta) \geq \ell_0-\eta/2.
 \end{align*}
Thus, for all $n > N$ and $\theta \in [0, \delta)$, we have
$$ \ell_0-\eta/2 \leq f_n(\theta) \leq \ell_0 + \eta.$$
Taking the limit as $n \to \infty$, we conclude that for $\theta \in [0, \delta)$,
$$\ell_0-\eta/2 \leq f(\theta) \leq \ell_0 + \eta,$$
implying continuity at $0$. Continuity at $1$ follows similarly.

\section{Proofs for Section \ref{section: part a}}

\subsection{Proof of Lemma \ref{lemma: alex}}\label{proof: lemma: alex}
\noindent To show concavity of $a_n(\cdot)$, note that all we need to prove is that
\begin{equation}
a_n\left( \frac{j}{n} \right) \geq \frac{ a_n\left(\frac{j-1}{n} \right) + a_n\left( \frac{j+1}{n}\right) }{2} \text{~~for all~~} 1 \leq j \leq n-1,
\end{equation}
as $a_n$ is a linear interpolation of the values at $\frac{j}{n}$. This is equivalent to proving
\begin{equation}
\mu_n(j)^2 \geq \mu_n(j-1)\mu_n(j+1) \text{~~for all~~} 1 \leq j \leq n-1.
\end{equation}
This is an easy application of the Alexandrov-Fenchel inequalities for mixed volumes. For a proof we refer to McMullen \cite{mcmullen1991inequalities}, where in fact the author obtains 
$$\mu_n(j)^2 \geq \frac{j+1}{j}\mu_n(j-1)\mu_n(j+1).$$

\subsection{Proof of Lemma \ref{lemma: sn linearized}}\label{proof: lemma: sn linearized}

Let $\eta > 0$ be given. The function $-(\Lambda^\epsilon)^*$, being continuous on the bounded interval $[0,1]$, is uniformly continuous. The same holds over the closed set $\cI \subset (0,1)$. Choose $\delta >0$ such that
$$|(\Lambda^\epsilon)^*(x) -(\Lambda^\epsilon)^*(y)| < \eta, \text{~~whenever~~} |x-y| < \delta, \text{ and } x, y \in \cI$$
Choose $N_0 > 1/(\delta/3)$, and divide the interval $\cI$ into the the $N_0$ equally sized closed intervals $I_j$, for $1 \leq j \leq N_0$. Let the interval $I_j := [c_{j-1}, c_j]$, where $[c_0, c_{N_0}] = \cI$. Without loss of generality, let $\theta_0$ lie in the interior of the $k$-th interval (we can always choose a different value of $N_0$ to make sure $\theta_0$ does not coincide with the boundary points $c_{k-1}$ or $c_k$). Thus,
$$c_{k-1}  < \theta_0 < c_{k}.$$
Theorem \ref{thm: lambda Tn} along with the continuity of $-(\Lambda^\epsilon)^*$ imply that
\begin{equation}
\lim_{n \to \infty} \frac{1}{n} \log \mu^\epsilon_n(I_j) = \sup_{\theta \in I_j} -(\Lambda^\epsilon)^*(\theta).
\end{equation}
For $n > 2/\min \left( \theta_0 - c_{k-1}, c_k - \theta_0\right)$, there exists an $i$ such that
\begin{equation}
c_{k-1} < \frac{i}{n} < \theta_0 < \frac{i+1}{n} < c_k.
\end{equation}
Thus for some $\lambda > 0$, we can write
\begin{equation}
a_n^\epsilon(\theta_0) = \lambda\frac{1}{n}\log \mu^\epsilon_{n/n}(i/n) + (1-\lambda)\frac{1}{n} \log\mu^\epsilon_{n/n}((i+1)/n),
\end{equation}
and obtain the inequality 
\begin{align}
a^\epsilon_n(\theta_0) &= \lambda\frac{1}{n}\log \mu^\epsilon_{n/n}(i/n) + (1-\lambda)\frac{1}{n} \log\mu^\epsilon_{n/n}((i+1)/n)\\
&\leq \max\left(\frac{1}{n}\log \mu^\epsilon_{n/n}(i/n), \frac{1}{n} \log\mu^\epsilon_{n/n}((i+1)/n)\right)\\
&\leq \frac{1}{n} \log \mu^\epsilon_{n/n}(I_k).
\end{align}
Thus we have the upper bound
\begin{align}\label{eq: an ub}
\limsup_n  a^\epsilon_n(\theta_0) &\leq \lim_{n\to \infty} \frac{1}{n} \log \mu^\epsilon_{n/n}(I_k)\\
 &= \sup_{\theta \in I_k} -(\Lambda^\epsilon)^*(\theta)\\
 &\stackrel{(a)}\leq -(\Lambda^\epsilon)^*(\theta_0) + \eta,
\end{align}\\
where $(a)$ follows from the choice of $N_0$ and uniform continuity of $-(\Lambda^\epsilon)^*$.\\

\noindent Define $$\hat \theta_n(j) = \arg \sup_{\frac{i}{n} \text{~s.t.~} \frac{i}{n} \in I_j} \mu^\epsilon_{n/n}\left(\frac{i}{n}\right).$$ As 
$$\mu^\epsilon_{n/n}(\hat \theta_n(j)) \leq \mu^\epsilon_{n/n}(I_j) \leq \left(n|I_j|+2\right)\mu^\epsilon_{n/n}(\hat \theta_n(j)) \leq n\mu^\epsilon_{n/n}(\hat \theta_n(j)),$$ it is easy to see that 
\begin{align}
\lim_{n\to\infty}\frac{1}{n} \log \mu^\epsilon_{n/n}(\hat \theta_n(j)) &= \lim_{n \to \infty} \frac{1}{n} \log \mu^\epsilon_{n/n}(I_j)\\
&= \sup_{\theta \in I_j} -(\Lambda^\epsilon)^*(\theta).
\end{align}
Note that since $a^\epsilon_n$ is obtained by a linear interpolation, we have
$$\sup_{\theta \in I_j}  a^\epsilon_n(\theta) \geq \frac{1}{n} \log \mu^\epsilon_{n/n}(\hat\theta_n(j)).$$
This implies that for the intervals $I_{k-1}$ and $I_{k+1}$,
\begin{align}
\liminf_{n\to \infty} \left[\sup_{\theta \in I_{k-1}}  a^\epsilon_n(\theta)\right] \geq \sup_{\theta \in I_{k-1}} -(\Lambda^\epsilon)^*(\theta) \geq -(\Lambda^\epsilon)^*(\theta_0) - \eta\\
\liminf_{n\to \infty} \left[\sup_{\theta \in I_{k+1}} a^\epsilon_n(\theta)\right] \geq \sup_{\theta \in I_{k+1}} -(\Lambda^\epsilon)^*(\theta) \geq -(\Lambda^\epsilon)^*(\theta_0) - \eta.
\end{align}
Since $ a^\epsilon_n(\theta)$ is concave, this implies 
\begin{align}
 a^\epsilon_n(\theta_0) &\geq \min (\sup_{\theta \in I_{k-1}}  a^\epsilon_n(\theta), \sup_{\theta \in I_{k+1}}  a^\epsilon_n(\theta)).
\end{align}
Taking the $\liminf$ on both sides,
\begin{align}\label{eq: an lb}
\liminf_{n \to \infty}  a^\epsilon_n(\theta_0) \geq -(\Lambda^\epsilon)^*(\theta_0) - \eta.
\end{align}
Inequalities \eqref{eq: an ub} and \eqref{eq: an lb} prove the pointwise convergence of $a^\epsilon_n(\theta_0)$ to $-(\Lambda^\epsilon)^*(\theta_0)$. We then use the following lemma to conclude uniform convergence:
\begin{lemma}\label{lemma: convex}
Let $\{f_n\}$ be a sequence of continuous convex functions which converge pointwise to a continuous function $f$ on an interval $[a, b]$. Then $f_n$ converge to $f$ uniformly.
\end{lemma}
\begin{proof}
Let $\epsilon > 0$. We'll show that there exists a large enough $N$ such that for all $n > N$, $||f_n-f||_\infty < \epsilon$.
\smallskip

The function $f$ is continuous on a compact set, and therefore is uniformly continuous. Choose a $\delta > 0$ such that $|f(x)- f(y)| < \epsilon/10$ for $|x-y| < \delta$. Let $M$ be such that $(b-a)/M < \delta$. We divide the interval $[a,b]$ into $M$ intervals, whose endpoints are equidistant. We denote them by $a = \alpha_0 < \alpha_1 < \cdots < \alpha_M = b$. Since $f_n(\alpha_i) \to f(\alpha_i)$, there exists a $N_i$ such that for all $n > N_i$, $|f_n(\alpha_i) - f(\alpha_i)| < \epsilon/10$. Choose $N = \max(M, N_0, \cdots, N_M)$.
\smallskip

Consider an $x \in (\alpha_i, \alpha_{i+1})$ for some $0 \leq i < M$, and let $n>N$. Using uniform continuity of $f$, we have
\begin{equation}\label{eq: bound on fx}
f(\alpha_i) -\epsilon/10 < f(x) < f(\alpha_i) + \epsilon/10.
\end{equation}
Further, we also have
\begin{align*}
f_n(\alpha_i) &\leq f(\alpha_i) + \epsilon/10~(\text{by pointwise convergence at } \alpha_i),\\
f_n(\alpha_{i+1}) &\leq f(\alpha_{i+1}) + \epsilon/10~(\text{by pointwise convergence at } \alpha_{i+1}),\\
&\leq f(\alpha_i) + 2\epsilon/10~ (\text{by uniform continuity of } f).
\end{align*}
Convexity of $f_n$ implies
\begin{equation}\label{eq: lower bound on fnx}
f_n(x) < \max(f_n(\alpha_i), f_n(\alpha_{i+1})) < f(\alpha_i) + 2\epsilon/10.
\end{equation}
Combining part of equation (\ref{eq: bound on fx}) and equation (\ref{eq: lower bound on fnx}), we obtain
\begin{equation}\label{eq: first half}
f_n(x) - f(x) < 3\epsilon/10.
\end{equation}
We'll now try to lower-bound $f_n(x)$. First consider the case when $i \geq 1$. In this case we have 
$$\alpha_{i-1} < \alpha_i < x < \alpha_{i+1}.$$
We write $\alpha_i$ as a linear combination of $x$ and $\alpha_{i-1}$, and use the convexity of $f_n$ to arrive at
\begin{align*}
&f_n(\alpha_i) \leq \frac{\alpha_i - \alpha_{i-1}}{x - \alpha_{i-1}}f_n(x) + \frac{x-\alpha_i}{x-\alpha_{i-1}}f_n(\alpha_{i-1}).
\end{align*}
This implies
\begin{align*}
 \frac{x-\alpha_{i-1}}{\alpha_i - \alpha_{i-1}}f_n(\alpha_i) - \frac{x-\alpha_i}{\alpha_i - \alpha_{i-1}}f_n(\alpha_{i-1}) \leq f_n(x).
\end{align*}
Taking the infimum of the left side, we get
\begin{align*}
\inf_{u \in (\alpha_i, \alpha_{i+1})} \frac{u-\alpha_{i-1}}{\alpha_i - \alpha_{i-1}}f_n(\alpha_i) - \frac{u-\alpha_i}{\alpha_i - \alpha_{i-1}}f_n(\alpha_{i-1}) \leq f_n(x).
\end{align*}
Note that since the LHS is linear in $x$, the infimum occurs at one of the endpoints of the interval, $\alpha_i$ or $\alpha_{i+1}$. Substituting, we get
\begin{align}
f_n(x) &\geq \min\left( f_n(\alpha_i), 2f_n(\alpha_i) - f_n(\alpha_{i-1})\right) \nonumber\\
&\geq \min(f(\alpha_i)-\epsilon/10, 2(f(\alpha_i)-\epsilon/10) - f(\alpha_{i-1}) - \epsilon/10)~\nonumber \\
&\geq \min(f(\alpha_i)-\epsilon/10, 2f(\alpha_i) - f(\alpha_{i-1}) - 3\epsilon/10)~\nonumber \\
&\geq \min(f(\alpha_i)-\epsilon/10, 2f(\alpha_i) - f(\alpha_i) - \epsilon/10 - 3\epsilon/10)\nonumber \\
&= f(\alpha_i) - 4\epsilon/10. \label{eq: upper bound on fnx}
\end{align}
Combining inequality (\ref{eq: upper bound on fnx}) with a part of inequality (\ref{eq: bound on fx}), we have
\begin{equation}\label{eq: second half}
f_n(x) - f(x) > -5\epsilon/10.
\end{equation}
Combining (\ref{eq: first half}) and (\ref{eq: second half}) we conclude that for all $x \in (\alpha_1, \alpha_M)$, and for all $n >N$, 
\begin{equation}
|f_n(x) - f(x)| < \epsilon/2.
\end{equation}
Now let $x \in (\alpha_0, \alpha_1)$. We express $\alpha_1$ as a linear combination of $x$ and $\alpha_2$ and follows the steps as above to establish \eqref{eq: second half} for $x \in (\alpha_0, \alpha_1)$. This shows that for all $x \in [a, b]$, $||f_n(x) - f(x)|| < \epsilon/2$ for all $n > N$, and concludes the proof.
\end{proof}


\section{Alternate definitions of typical sets}\label{appendix: alternate}

\begin{theorem}\label{thm: gen typ}
Let $X$ be a non-uniform log-concave random variable. For $\epsilon > 0$, define the following sequences of sets
\begin{align}
\overline \cT^\epsilon_n = \{x^n \mid p_{X^n}(x^n) \geq \exp(-n(h(X)+\epsilon))\}\\
\underline \cT^\epsilon_n = \{x^n \mid p_{X^n}(x^n) \geq \exp(-n(h(X)-\epsilon))\}.
\end{align}
Let $\{T_n\}$ be any sequence of compact convex sets. Suppose that for any $\epsilon >0$ there exists an $N(\epsilon)$ such that for all $n > N(\epsilon)$ the following  inclusion holds:

\begin{equation}
\underline \cT^\epsilon_n  \subseteq \cT_n \subseteq \overline \cT^\epsilon_n. 
\end{equation}
Then the following equality holds for all $\theta \in [0,1]$:
\begin{equation}
\lim_{n \to \infty} \frac{1}{n} \log V_{\lfloor n\theta \rfloor}(T_n) = h_X(\theta)
\end{equation}
\end{theorem}
\begin{remark}
We exclude uniform random variables because $\underline \cT^\epsilon_n = \phi$ for all $n$ and all $\epsilon > 0$.
\end{remark}

\begin{proof}
By monotonicity of intrinsic volumes, we note that  it is enough to show 
\begin{equation}
\lim_{n \to \infty} \frac{1}{n} \log V_{\lfloor n\theta \rfloor}(\underline \cT^\epsilon_n) = h_X(\theta).
\end{equation}
Using the same proof idea as in Lemma \ref{lemma: superconv Tn}, we see that$\{\underline \cT^\epsilon_n \}$ is a super-multiplicative sequence. Using the non-uniformity of $X$, we observe that for all small enough $\epsilon$ the set $\underline \cT^\epsilon_n \neq \phi$ for all $n$. Furthermore, noting that $\underline \cT^\epsilon_n \subseteq \overline \cT^\epsilon_n $ we see that the intrinsic volumes of $\{\underline \cT^\epsilon_n \}$ constitute a proper super-convolutive sequence. Thus, the limit function for the intrinsic volumes of $\{\underline \cT^\epsilon_n\}$ exists, although it may not be continuous at 1 (continuity at 0 follows since it is bounded above $h_X$. For ease of notation we refer to the limit functions of the intrinsic volumes of $\{\underline \cT^\epsilon_n\}$ and $\{\overline \cT^\epsilon_n\}$ by $\underline \cG^\epsilon$ and $\overline \cG^\epsilon$. Our goal is to show that $\underline \cG^\epsilon$ and $\overline \cG^\epsilon$, cannot differ by too much. We show that it is possible to ``bloat" $\{\underline \cT^\epsilon_n\}$ by a small fraction, so that the bloated set will contain $\{\overline \cT^\epsilon_n\}$.

Let $p_X(x) = \exp(-\Phi(x))$, and assume without loss of generality that $\Phi$ achieves its minimum at 0. Note that $h(X) = E \Phi(X)$, and thus $h(X) \geq \Phi(0)$. This inequality is strict when $X$ is not uniform. Assume $\epsilon < h(X) - \Phi(0)$. The two sequences of sets may be equivalently described as
\begin{align}
\overline \cT^\epsilon_n = \{x^n \mid \sum \Phi(x_i) \leq  n(h(X)+\epsilon)\}\\
\underline \cT^\epsilon_n = \{x^n \mid \sum \Phi(x_i) \leq  n(h(X)-\epsilon)\}.
\end{align}
For $\alpha > 0$, consider the map that takes $x^n \in \mathbb R$ to  $(1+\alpha)x^n$. Let $x^n$ be a point on the boundary of $\underline \cT^\epsilon_n$ satisfying $\sum \Phi(x_i) = n(h(X)-\epsilon)$. We have the inequalities
\begin{align*}
\sum_{i=1}^n \Phi(x_i(1+\alpha)) &\stackrel{(a)}\geq \sum_{i = 1}^n \Phi(x_i) + \sum_{i=1}^n \alpha (\Phi(x_i) - \Phi(0))\\
&= n(1+\alpha)(h(X)-\epsilon) - n\Phi(0).
\end{align*}
where $(a)$ follows from the convexity of $\Phi$. Now for a choice of $\alpha = \frac{2\epsilon}{h(X)-\epsilon-\Phi(0)}$, we will have
$$\sum_{i=1}^n \Phi(x_i(1+\alpha)) \geq n(h(X) + \epsilon),$$
that is, $x^n(1+\alpha) \notin \overline \cT^\epsilon_n$. Note that $\alpha \to 0$ as $\epsilon \to 0$. Thus, for  this choice of $\alpha$ we must have $$\overline \cT^\epsilon_n \subseteq (1+\alpha)\underline \cT^\epsilon_n.$$ This implies
$$\overline \cG^\epsilon \leq \log (1+\alpha) + \underline \cG^\epsilon.$$
We also have $\underline \cG^\epsilon \leq \overline \cG^\epsilon$, giving
$$\overline \cG^\epsilon \leq \log (1+\alpha) + \underline \cG^\epsilon \leq \log (1+\alpha) + \overline \cG^\epsilon.$$
Since and $\lim_{\epsilon \to 0} \overline \cG^\epsilon = h_X$, we can take the limit as $\epsilon \to 0$ to conclude that 
$$\lim_{\epsilon \to 0} \underline \cG^\epsilon = h_X.$$
This concludes the proof.
\end{proof}

\section{G\"{a}rtner-Ellis theorem}\label{appendix: gartnerellis}
\begin{theorem}[G\"{a}rtner-Ellis theorem]\label{thm: gartner ellis}
Consider a sequence of random vectors $Z_n \in \mathbb R^d$, where $Z_n$ possess the law $\mu_n$ and the logarithmic moment generating function
$$\Lambda_n(\lambda) := \log E\left[\exp\left(\sum_{i=1}^d \lambda_iZ_i\right)\right].$$ We assume the following:
\begin{itemize}
\item[$(*)$]
 For each $\lambda\in \mathbb R^d$, the limit
$$\Lambda(\lambda) := \lim_{n \to \infty} \frac{1}{n} \Lambda_n(n\lambda)$$
exists as an extended a real number. Further the origin belongs to the interior $\cD_\Lambda := \{\lambda \in \mathbb R^d~|~\Lambda(\lambda) < \infty\}$.
\end{itemize}
 Let $\Lambda^*$ be the convex conjugate of $\lambda$ with $\cD_{\Lambda^*} = \{x \in \mathbb R^d~|~ \Lambda^*(x) < \infty\}$. When assumption $\textbf{(*)}$ holds, the following are satisfied:
 \begin{itemize}
 \item[1.] For any closed set $I$,
 $$\limsup_{n \to \infty} \frac{1}{n} \log \mu_n(I) \leq -\inf_{x \in I} \Lambda^*(x).$$
 \item[2.] For any open set $F$,
 $$\liminf_{n \to \infty} \frac{1}{n} \log \mu_n(F) \geq  -\inf_{x \in F \cap \cF} \Lambda^*(x),$$
 where $\cF$ is the set of exposed points of $\Lambda^*$ whose exposing hyperplane belongs to the interior of $\cD_{\Lambda}$.
 \item[3.]
 If $\Lambda$ is an essentially smooth, lower semicontinuous function, then the large deviations principle holds with a good rate function $\Lambda^*$.
 \end{itemize}
\end{theorem}
\begin{remark}\label{remark: gartnerellis derivative}
For definitions of \emph{exposed points}, \emph{essentially smooth functions} and \emph{good rate function}, we refer to Section $2.3$ of \cite{dembo1998large}. For our purpose, it is enough to know that if $\Lambda$ is differentiable on $\cD_{\Lambda} = \mathbb R^d$, then is it essentially smooth.
\end{remark}

\end{appendix}
\end{document}